%% file: sampling_variability_forecast_combinations.tex
\pdfoutput=1

\documentclass[12pt]{article}

\newcommand{\blind}{0}

\usepackage{amsfonts}
\usepackage{amsmath}
\usepackage{amssymb}
\usepackage{amsthm}
\usepackage{booktabs}
\usepackage{caption}
\usepackage{colortbl}
\usepackage{graphicx}
\usepackage[colorlinks, linkcolor=blue, citecolor=blue, urlcolor=blue, bookmarks=false, hyperfootnotes=false]{hyperref}
\usepackage{mathtools}
\usepackage{multibib}
\usepackage{multirow}
\usepackage{natbib}
\usepackage{setspace}
\usepackage[compact]{titlesec}
\usepackage{titling}

\definecolor{Gray}{gray}{0.9}

\newtheorem{theorem}{Theorem}

\newtheorem{assumption}{Assumption}
\theoremstyle{definition}

\theoremstyle{remark}
\newtheorem{remark}{Remark}
\newtheorem*{remark*}{Remark}

\DeclareMathOperator*{\argmax}{\arg\!\max}

\renewcommand{\appendixname}{Appendix}
\setcitestyle{authoryear,open={(},close={)},aysep={}}
\newcites{supp}{Supplementary References}

\begin{document}

\def\spacingset#1{\renewcommand{\baselinestretch}%
{#1}\small\normalsize} \spacingset{1}

\newcommand{\papertitle}{\bf The Impact of Sampling Variability on Estimated Combinations of Distributional Forecasts\protect\footnotemark[1]}
\newcommand{\paperauthors}{Ryan Zischke\protect\footnotemark[2] \protect\footnotemark[3] \protect\footnotemark[4], Gael M. Martin\protect\footnotemark[2], David T. Frazier\protect\footnotemark[2], D. S. Poskitt\protect\footnotemark[2]}
\newcommand{\funding}{This research has been supported by Australian Research Council (ARC) Discovery Grant DP200101414. Frazier was also supported by ARC Early Career Researcher Award DE200101070; and Martin and Frazier were provided support by the ARC Centre of Excellence in Mathematics and Statistics.}
\newcommand{\ebsaffiliation}{Department of Econometrics and Business Statistics, Monash University.}
\newcommand{\mdaffiliation}{Methodology Division, Australian Bureau of Statistics. The views expressed in this paper are those of the authors alone, and do not in any way represent the Methodology Division, the Australian Bureau of Statistics, the Australian Public Service, or the Australian Government.}
\newcommand{\correspondingauthor}{Corresponding author \\ \indent\ \ \textit{E-mail address:} ryan.zischke@monash.edu.}

\if0\blind
{
  \renewcommand{\thefootnote}{\fnsymbol{footnote}}
  \title{\papertitle}
  \date{June 6, 2022}
  \author{\paperauthors}
  \maketitle
  \footnotetext[1]{\funding}
  \footnotetext[2]{\ebsaffiliation}
  \footnotetext[3]{\mdaffiliation}
  \footnotetext[4]{\correspondingauthor}
  \setcounter{footnote}{0}
  \renewcommand{\thefootnote}{\arabic{footnote}}
} \fi

\if1\blind
{
  \spacingset{1.5}
  \bigskip
  \bigskip
  \bigskip
  \begin{center}
    {\Large\bf The Impact of Sampling Variability on Estimated Combinations of Distributional Forecasts}
  \end{center}
  \smallskip
  \spacingset{1}
} \fi

\begin{abstract}
We investigate the performance and sampling variability of estimated forecast combinations, with particular attention given to the combination of forecast \textit{distributions}. Unknown parameters in the forecast combination are optimized according to criterion functions based on proper scoring rules, which are chosen to reward the form of forecast accuracy that matters for the problem at hand, and forecast performance is measured using the out-of-sample expectation of said scoring rule. Our results provide novel insights into the behavior of estimated forecast combinations. Firstly, we show that, {asymptotically}, the sampling variability in the performance of standard forecast combinations is determined {solely} by estimation of the constituent models, with estimation of the combination weights contributing no sampling variability whatsoever, at first order. Secondly, we show that, if computationally feasible, forecast combinations produced in a single step -- in which the constituent model and combination function parameters are estimated jointly -- have superior {predictive }accuracy and lower sampling variability than standard forecast combinations -- where constituent model and combination function parameters are estimated in two steps. These theoretical insights are demonstrated numerically, both in simulation settings and in an extensive empirical illustration using a time series of S\&P500 returns.
\end{abstract}

 \if0\blind { \\ } \fi

\if0\blind
{
} \fi

\spacingset{1.5}

\section{Introduction}

Since the publication of the seminal papers by \cite{Stone1961} and \cite{Bates1969}, the forecasting literature has seen an explosion of interest in the production of point and/or distributional forecasts via weighted combinations of forecasts from distinct models (see \citealp{Timmermann2006}, \citealp{Aastveit2019}, and \citealp{Wang2022} for relevant reviews). Forecast combinations (of both types) have attracted such attention in part because of their highly competitive performance out-of-sample (\citealp{Makridakis2018, Makridakis2020}; \citealp{Thorey2018}; \citealp{Wang2018}; \citealp{Taylor2020}). At the same time, however, there remains considerable disagreement about the relative merits of different combination approaches, including the role played by estimation error in the performance of such methods.

Two popular methods of producing \textit{distributional} forecast combinations, for instance, are the linear opinion pool (\citealp{Stone1961}; \citealp{Hall2007}; \citealp{Geweke2011}; \citealp{Opschoor2017}; \citealp{Martin2021}) and the beta-transformed linear opinion pool (\citealp{Ranjan2010}; \citealp{Gneiting2013}; \citealp{Satopaeae2014}; \citealp{Baran2018}). One popular method of producing \textit{point} forecast combinations is by taking a weighted average of the constituent forecasts (\citealp{Bates1969}; \citealp{Stock2004}; \citealp{Timmermann2006}; \citealp{Smith2009}; \citealp{Claeskens2016}). All three of these approaches are examples of \textit{combination functions}, which are (often parameterized) functions that produce a single forecast using forecasts from constituent models, for each future time point.

Given a combination function, one is confronted with the choice of whether to estimate the weight given to each constituent model, or to fix the weights at given values, usually a vector of values that assigns equal weight to each model. Experimentation with equally weighted combinations, and those that are ``optimized'' according to some criterion, has revealed that choosing the weights to optimize some reward function does not necessarily lead to better out-of-sample performance in practice, as measured by that reward function \citep{Clemen1989}; a finding known colloquially as the ``forecast combination puzzle''. Empirical studies have provided support for both {possible }positions. {The analysis of \cite{Stock2004} and \cite{Smith2009} suggest that estimating forecast combination weights has little impact on the accuracy of forecast combinations, with the equally weighted combinations marginally outperforming combinations with optimized weights. On the other hand, there is also empirical evidence to suggest that optimal weights can outperform the equally weighted combination (e.g.\ \citealp{Genre2013}; \citealp{Hsiao2014}; \citealp{Martin2021}).

Investigations have, to our knowledge, only focused on the impact of combination weight estimation (e.g.\ \citealp{Smith2009}; \citealp{Claeskens2016}; \citealp{Elliott2017}; \citealp{Chan2018}), with the impact of sampling variability due to the estimation of the constituent model parameters largely being neglected. In this paper, we explore the impact of sampling variability in forecast combinations, comprising variability in estimates of the constituent model \textit{and} combination parameters, on the accuracy of combinations of \textit{distributional forecasts}. Herein, we measure the accuracy of distributional forecast combinations through expected (proper) scoring rules (\citealp{Gneiting2007}), and rigorously analyze how sampling variability in the full vector of estimated parameters impacts (distributional) forecast accuracy of estimated combinations.

Our first key finding is that, in general contexts, and under weak regularity conditions, estimation of the combination weights imparts no bias or variability into the performance of standard forecast combinations, at least as far as first order asymptotic behavior is concerned. This finding lies in opposition to the analysis of \cite{Claeskens2016}, who show that in certain stylized scenarios estimation of the weights imparts additional bias and variance into the forecasts, but fits squarely within the findings of \cite{Smith2009}, who argue that the {failure of the optimized combination to perform best} is due to finite-sample error. This key result is used to demonstrate that the sampling variability of our chosen forecast accuracy measure depends only on the variability of the estimated constituent model parameters, and underscores the significant role played by the estimated constituent models used in forecast combinations. Since the variability in the estimated constituent models is often much larger than the sampling variability of the estimated weights, our results suggest that not incorporating the variability of constituent model estimates within the analysis of forecast combinations may lead to conclusions that have little practical relevance.

The second key contribution made in this paper is to compare and contrast the behavior of forecast combinations that have been produced in the standard manner, against forecast combinations that are produced in a single step. Historically, forecast combinations have been produced by first producing (estimated) forecast distributions from the constituent models, and then estimating, or fixing, the corresponding combination weights. This two-step procedure allows us to view commonly applied forecast combination approaches as \textit{two-stage} extremum estimators (\citealp{Pagan1986}; \citealp{Newey1994}; \citealp{Frazier2017}), whereby the constituent models are estimated in the first stage, and the combination parameters are estimated in the second stage, conditional on the first stage estimates. In contrast, if the estimated weights and constituent model parameters were to be estimated in a single step, the resulting (vector) parameter estimator could be viewed as a \textit{one-stage} estimator, defined by a joint optimization program. Our results demonstrate that, when feasible, one-stage forecast combinations are preferable to standard two-stage forecast combinations. In particular, forecast combinations produced in a single step have higher accuracy, and lower sampling variability, than those produced in the standard two-stage manner, as measured by our chosen forecast performance measure.\footnote{\scriptsize We acknowledge the extensive Bayesian literature on combining forecast distributions (e.g.\ \citealp{Billio2013}; \citealp{Casarin2015}; \citealp{Casarin2016}; \citealp{Pettenuzzo2016}; \citealp{Aastveit2018}; \citealp{Bassetti2018}; \citealp{Bastuerk2019}; \citealp{Casarin2019}; \citealp{McAlinn2019}; \citealp{Loaiza-Maya2021}). The Bayesian treatment of this problem is, of course, fundamentally different from the frequentist approach adopted here.}

The paper proceeds as follows. In Section \ref{sec:accuracy} we introduce generic notation for the distributional forecast combinations we consider, and discuss how to assess the performance of these combinations. Section \ref{sec:implications} uses this framework to demonstrate that the sampling variability of standard (two-stage) forecast combination methods is entirely driven by the sampling variability in the constituent models, with no contribution from estimation of the weights; furthermore, we show that one-stage combination methods have superior forecasting performance relative to the standard forecast combination. In Section \ref{sec:montecarlo}, we illustrate all findings numerically in a simulation setting, using a distributional forecast combination based on two constituent models. Asymptotic results on forecast performance are shown to apply in a finite-sample empirical setting in Section \ref{sec:emp}, using a time series of daily S\&P500 returns from January 5th, 1988 to December 31st, 2021, with the dominance of the one-stage estimator seen to be robust to an increase in volatility. It is also shown that optimizing for forecast accuracy in the tails is of benefit during a period of high volatility, no matter what measure of out-of-sample forecast accuracy is used. Section \ref{sec:conclusion} concludes. Proofs for all \if0\blind{theoretical\ }\fi results can be found in the supplementary appendix\if0\blind{, and software for reproducing numerical results is available at \href{https://github.com/zisc/samp-var-comb}{https://github.com/zisc/samp-var-comb}}\fi.

The following notation is used throughout the paper and the supplementary appendix. Given variables $y_{1},y_{2},\ldots ,y_{t}$, we denote by $y_{1:t}$ the column vector $[y_{1}\ y_{2}\ \cdots \ y_{t}]^{\prime }$. For a sequence $a_{n}$ converging to zero, the terms $\mathcal{O}_{p}(a_{n})$ and $o_{p}(a_{n})$ are used to describe the convergence of a random variable relative to $a_{n}$; see \cite{VanDerVaart1998} for a textbook treatment. The symbols $\Rightarrow $ and $\overset{p}{\rightarrow }$ denote convergence in distribution and convergence in probability, respectively, and $\mathrm{plim}_n X_n$ denotes the probability limit as $n \to \infty$ of a random sequence $X_n$. A dot is used to represent \textquotedblleft the function whose argument replaces the dot\textquotedblright . For example, $f=g(\,\cdot \,,y)$ defines a function $f(x)=g(x,y)$ for some fixed $y$.

\section{Accuracy of Distributional Forecast Combinations\label{sec:accuracy}}

We first discuss {the use of scoring rules to} assess the accuracy of distributional forecasts (Section \ref{subsec:score}), then formalize the class of distributional forecast combinations (Section \ref{subsec:comb}), and discuss practical implementation details (Section \ref{subsec:produce}), before describing the precise manner in which we measure the forecast performance of estimated combinations (Section \ref{subsec:assess}).

\subsection{Scoring Rules\label{subsec:score}}

We follow \cite{Gneiting2007} and measure the accuracy of distributional forecasts using scoring rules. A \textit{scoring rule} (or \textit{score}) $S(F,y)$ measures the accuracy of a predictive cumulative distribution function (CDF) $F$, when the variable we are trying to predict, $Y$, achieves realization $Y=y$. All scoring rules in our analysis are taken to be positively orientated, so that $S(F,y)$ can be interpreted as the reward the forecaster achieves for quoting distribution $F$ when the realization $Y=y$ is observed.

Further, we consider that the scoring rules $S$ are \textit{strictly proper}: for a random variable $Y$, with true distribution function $H=\mathrm{Pr}(Y \leq \cdot \,)$, a scoring rule is proper if the expected score under $H$, $\mathbb{E}_{H}[S(F, Y)]=\int S(F,y)\mathrm{d}H$, is maximized at $F=H$; and the score is called \textit{strictly proper} if it is also maximized nowhere else. Throughout the remainder we take $S(\cdot,\cdot)$ to represent an arbitrary scoring rule that is strictly proper, and positively-oriented. See \cite{Gneiting2007} for a thorough introduction to use of scoring rules in measuring the accuracy of forecast distributions.

Arguably the most well-known scoring rule is the logarithmic score (log score) (\citealp{Hall2007}; \citealp{Geweke2011}), which is defined by 
\begin{equation}
S^{\mathrm{LS}}(F,y) = \log f(y),  \label{eqn:simsl}
\end{equation}
where $f$ is the probability density function (PDF) of the CDF $F$. Another popular scoring rule is the censored log score (\citealp{Diks2011}; \citealp{Opschoor2017}), which prioritizes accurate forecasts in a region of the support $B$, potentially at the expense of forecast accuracy on $B^c$, the region outside $B$. The censored log score is defined by 
\begin{equation}
S^{\mathrm{CS}_B}(F, y) = 
\begin{cases}
\log f(y) & y \in B \\ 
\log \left( \int_{B^c} f(y)\ \mathrm{d}y \right) & y \in B^c
\end{cases}.  \label{eqn:simscls}
\end{equation}
Both of these scoring rules feature in our numerical work in Sections \ref{sec:montecarlo} and \ref{sec:emp}.

\subsection{Defining Distributional Forecast Combinations\label{subsec:comb}}

In many cases, the forecaster entertains several possible models that can be used to predict a variable of interest. Rather than choosing a single model, she can combine several models to form a distributional forecast combination: a distributional forecast combination is formed from the combination of several constituent forecast distributions. Distributional forecast combinations are thus formed from two pieces: the (predictive) distributions from the constituent models, and the method by which the distributions are combined, i.e.\ the combination function.

Consider $Y_{1},\dots ,Y_{t}$ generated from the probability triplet $(\Omega ,\mathcal{A},H)$; given $Y_{1:t-1}$, our goal is to produce a forecast for the distribution of the random variable $Y_{t}|Y_{1:t-1}$. Let $F_{1}(\,\cdot \mid Y_{1:t-1};\gamma _{1})$ denote a predictive CDF for $Y_{t}|Y_{1:t-1}$ that depends on unknown parameters $\gamma _{1}\in \Gamma_{1}\subseteq \mathbb{R}^{d_{\gamma _{1}}}$.\footnote{\scriptsize To apply the results of the current paper to $h$-step-ahead forecasts, simply replace \textquotedblleft $\scriptstyle{1:t-1}$\textquotedblright\ with \textquotedblleft $\scriptstyle{1:t-h}$\textquotedblright\ throughout.} We use the summary notation $F_{1,t}^{\gamma _{1}} \equiv F_{1}(\,\cdot \mid Y_{1:t-1};\gamma _{1})$ hereafter.

While the practitioner can use $F_{1,t}^{\gamma _{1}}$ to produce a forecast distribution for $Y_{t}|Y_{1:t-1}$, in the majority of forecasting settings there exists a collection of $K$ possible models that can be used to produce forecast distributions. We suppose that each of these forecast distributions is indexed by their constituent model parameters, $\gamma _{1},\dots ,\gamma_{K}$, with $\gamma _{j}\in \Gamma _{j}\subseteq \mathbb{R}^{d_{\gamma_{j}}},j=1,2,\ldots ,K$, which gives the practitioner access to $K$ forecast distributions $\{F_{1,t}^{\gamma _{1}},\dots ,F_{K,t}^{\gamma _{K}}\}$ to choose amongst. Rather than adopt just one of these models, we can consider combining them to produce a single forecast distribution. In this way, we can stack the unknown parameters into a single vector that contains all the parameters from the constituent models: $\gamma \in \Gamma \coloneqq\Gamma_{1}\times \Gamma _{2}\times \cdots \times \Gamma _{K}$.

Following \cite{Gneiting2013}, we consider a combination function $Q_{\eta}$ indexed by the parameter vector $\eta \in \mathcal{E}\subseteq \mathbb{R}^{d_{\eta }}$. Collecting the unknown parameters as $\theta = [\eta^{\prime}\ \gamma^{\prime}]^{\prime} \in \Theta \coloneqq \mathcal{E} \times \Gamma$, at each time-period $t$, the combination function $Q_{\eta }$ produces a single one-step-ahead predictive distribution $F_{c,t}^{\theta} \equiv F_{c}(\,\cdot \mid Y_{1:t-1};\eta, \gamma)$, where $F_{c}$ denotes the composition of $Q_{\eta }$ and the $K$ predictive distributions $F_{j,t}^{\gamma _{j}},\ j=1,2,\ldots ,K$ from the constituent models under consideration. As noted earlier, and with selected references, common choices for the combination function include the linear pool and the beta-transformed linear pool, with the vector $\eta$ comprising the familiar set of weights on the unit simplex in the former case. At no stage do we assume that the true data generating process (DGP) is spanned by the forecast combination.

\subsection{Producing Forecast Combinations\label{subsec:produce}}

The single predictive CDF $F_{c,t}^{\theta}$, which is the composition of $Q_{\eta}$ and $F_{1,t}^{\gamma_{1}},\ldots,F_{K,t}^{\gamma_{K}}$, can be interpreted as a statistical model, $P_{\theta}$, indexed by an unknown parameter vector $\theta \in \Theta$, which we use to produce a predictive distribution for $Y_{t}$ given a random sample $Y_{1}, Y_{2}, \ldots, Y_{t-1}$. With this construction, we aim to choose parameter estimates $\hat{\theta}_{n} \equiv \hat{\theta}_{n}(Y_{1}, Y_{2}, \ldots, Y_{n})$ using the observed data so that $P_{\hat{\theta}_{n}}$ is, in some sense, a good approximation to the true unknown distribution of $Y_{t}|Y_{1:t-1}$.

A clear choice then is to follow \cite{Gneiting2007} (see also \citealp{Martin2021}), and produce $\hat{\theta}_{n}$ by maximizing the in-sample average scoring rule in which we wish to measure the accuracy of our predictive distributions: for $\mathcal{S}_{n}(\theta) = \frac{1}{n} \sum_{t=1}^{n} S(F_{c,t}^{\theta}, Y_{t})$, the unknown parameters can be estimated as 
\begin{equation}
\hat{\theta}_{n} \coloneqq \argmax_{\theta \in \Theta} \mathcal{S}_{n}(\theta). \label{eqn:est1s}
\end{equation}
Once $\hat{\theta}_{n}$ has been obtained, the forecast combination function $F_{c,n}^{\hat{\theta}_{n}}$ can be used to produce a forecast distribution for the random variable $Y_{n+1}$. Such an approach has been referred to as an ``optimal prediction approach'' by \cite{Gneiting2007} since maximizing the in-sample criterion we wish to evaluate should, \textit{a priori}, produce forecast distributions that attain high accuracy in that chosen score.

In the case of a predictive combination, however, the dimensionality of $\theta$ is often high. Hence, and in the spirit of the general forecast combination literature (e.g.\ \citealp{Hall2007}; \citealp{Geweke2011}; \citealp{Gneiting2013}), the most common approach to producing forecast combinations is to proceed in two steps (e.g.\ \citealp{Clemen1989}; \citealp{Stock2004}; \citealp{Geweke2011}; \citealp{Genre2013}; \citealp{Makridakis2020}; \citealp{Martin2021}); first, estimate the parameters of the constituent models, $\gamma_{1}, \ldots, \gamma_{K}$, then condition on these estimates to estimate the combination parameters, $\eta$.

More particularly, given our sample $Y_{1:n}$, we consider that the two-stage forecast distribution is produced using an estimator for $\theta$ that is constructed in two steps. First, each $\gamma_j$, $j = 1, 2, \ldots, K$ is estimated by maximizing the in-sample average scoring rule: $\tilde{\gamma}_{j,n} \coloneqq \argmax_{\gamma_j \in \Gamma_j} \frac{1}{n} \sum_{t = 1}^{n} S(F^{\gamma_j}_{j,t}, Y_{t})$, and we collect $\tilde{\gamma}_{j,n}$, $j = 1, \ldots, K$, as $\tilde{\gamma}_{n} \coloneqq [\tilde{\gamma}_{1,n}^{\prime}\ \cdots\ \tilde{\gamma}_{K,n}^{\prime}]^{\prime}$. In the second stage, the combination parameters $\eta$ are estimated via
\begin{equation}
\tilde{\eta}_n \coloneqq \argmax_{\eta \in \mathcal{E}} \mathcal{S}_n(\eta, \tilde{\gamma}_n),\quad \mathcal{S}_{n}(\theta) \equiv \mathcal{S}_{n}(\eta, \gamma) \coloneqq \frac{1}{n} \sum_{t = 1}^n S(F^{\theta}_{c,t}, Y_{t}),
\label{eqn:est2s1}
\end{equation}
and the estimated parameters are stacked as 
\begin{equation}
\tilde{\theta}_n \coloneqq [\tilde{\eta}_n^{\prime}\ \tilde{\gamma}_n^{\prime}]^{\prime}.  \label{eqn:est2s2}
\end{equation}

\begin{remark}
For treatments of two-stage estimation, see e.g.\ \cite{Pagan1986}; \cite{Newey1994}; \cite{Frazier2017}. To paraphrase \cite{Pagan1986}, in the context of predictive combinations we can view the parameters that underlie the constituent predictive models as nuisance parameters and note that \textquotedblleft estimation would generally be easy if the nuisance parameters were known.\textquotedblright\ In this setting, our chosen strategy for dealing with these nuisance parameters is to replace each of them by some value that is estimated from the data. As we shall see later, this seemingly innocuous practice increases the sampling variability of the forecast performance of our estimated model $P_{\tilde{\theta}_{n}}$.
\end{remark}

We will refer to forecast distributions produced through a single estimation step as one-stage forecast combinations in order to distinguish them from standard (two-stage) forecast combinations.

\subsection{Assessing Accuracy of Forecast Combinations\label{subsec:assess}}

In the case of single models, e.g., $F_{j,t}^{\gamma _{j}}$, accuracy of forecast combinations can be measured through the limiting expected average score 
\begin{equation}
\lim_{n\rightarrow \infty }\mathbb{E}_{H}\left[ \frac{1}{n} \sum_{t=1}^{n}S(F_{j,t}^{\gamma _{j}},Y_{t})\right].  \label{eqn:indiv}
\end{equation}
In the case of forecast combinations, accuracy of forecast distributions can be measured via 
\begin{equation}
\mathcal{S}_{0}(\theta) \coloneqq \lim_{n\rightarrow \infty} \mathbb{E}_{H} \left[ \mathcal{S}_{n}(\theta) \right],\ \mathcal{S}_{0}(\theta) \equiv \mathcal{S}_{0}(\eta, \gamma),  \label{eqn:S0}
\end{equation}
with $\mathcal{S}_{n}(\theta )$ as defined in \eqref{eqn:est2s1}, and where we remind the reader that $\mathbb{E}_{H}$ denotes expectation under the true, unknown, distribution $H$ of $Y_{1}, \dots, Y_{n}, \dots$. For the remainder of the paper we drop the subscript for brevity, and note that all expectations and variances are to be taken with respect to the true distribution $H$.

Rather than attempting to draw conclusions using a conventional in-sample versus out-of-sample split, we instead consider the performance of forecast combinations through the theoretically important construct of $\mathcal{S}_{0}(\theta)$.\footnote{While, strictly speaking, $\mathcal{S}_{0}$ is infeasible to construct in practice, it can be estimated consistently (as $n \rightarrow \infty $) using a sequence of out-of-sample forecast evaluations.} We choose this approach for measuring forecast performance for two main reasons: firstly, it obviates the need to choose the {in-sample/out-of-sample} evaluation mechanism, which has important ramifications for the performance of forecast evaluation and {which}, depending on the sample splitting, may ultimately obscure the sampling variability of the forecast combinations; secondly, this approach will allow us to easily link the sampling variability of different forecast combination schemes to their performance.

\section{Implications for Forecast Performance\label{sec:implications}}

In this section, we show that one-stage combinations have superior forecast performance over standard two-stage combinations, and that this performance measure also has less sampling variability in the one-stage case. For the sake of brevity, and to keep technical details to a minimum, we only present the key results. Regularity conditions, technical discussion of said conditions, and proofs of all stated results are collected in the supplementary appendix.

As discussed in Section \ref{subsec:assess}, we measure forecast performance via $\mathcal{S}_{0}(\theta)$. Therefore, the forecasting performance associated with the one- and two-stage forecast combinations can be assessed via $\mathcal{S}_{0}(\hat{\theta}_{n})$ and $\mathcal{S}_{0}(\tilde{\theta}_{n})$. To establish the behavior of these quantities, we first note that, under standard regularity conditions, see, e.g., \citet{Newey1994}, as $n \rightarrow \infty $, $\hat{\theta}_{n} \overset{p}{\rightarrow} \theta^{0} = [\eta^{0^{\prime}}\ \gamma^{0^{\prime}}]^{\prime}$, and $\tilde{\theta}_{n} \overset{p}{\rightarrow} \theta^{\star} = [\eta^{\star^{\prime}}\ \gamma^{\star^{\prime}}]^{\prime}$, with $\gamma^{\star} = [\gamma_{1}^{\star^{\prime}}\ \cdots\ \gamma_{K}^{\star^{\prime}}]^{\prime}$ where each $\gamma_{j}^{\star}$ maximizes the corresponding limiting expected average score in \eqref{eqn:indiv}, $\eta^{\star}$ maximizes $\mathcal{S}_{0}(\eta, \gamma^{\star})$, and $\theta^{0}$ maximizes $\mathcal{S}_{0}(\theta)$.}

As a measure of forecast performance, the out-of-sample expected average score is in the same spirit as the performance measures considered by \cite{West1996}, \cite{Hansen2005} and \cite{Giacomini2006}. However, unlike the above references, we are not concerned with testing the accuracy of different forecasting methods. Our main goal is to understand and document the impact of sampling variability on different forecast combination methods, and we are not interested in ascertaining which forecasting combination method provides superior forecasting performance; at least not within a conventional, i.e., formal, hypothesis testing framework.

We can summarize the first-order implications for the forecast performance of the one- and two-stage approaches as follows.

\begin{theorem}
\label{thm:one} If Assumptions \ref{ast:exp}-\ref{ast:dist} in \appendixname\ \ref{subsec:regularity} are satisfied, then the following results hold.

\smallskip

\noindent 1. With probability converging to one, as $n\rightarrow\infty$, $\mathcal{S}_0(\hat\theta_n) > \mathcal{S}_0(\tilde\theta_n)$.

\smallskip

\noindent 2. $\|\mathcal{S}_0(\tilde{\eta}_{n}, \tilde{\gamma}_{n}) - \mathcal{S}_0(\eta^{\star}, \tilde{\gamma}_{n}) \| = o_p(n^{-1/2})$.

\end{theorem}

\begin{remark}
\label{rmk:oracle} Theorem \ref{thm:one} has implications for the practical performance of forecast combinations. Firstly, Part 1.\ of Theorem \ref{thm:one} shows that the two-stage nature of standard forecast combinations ensures that, in any strictly proper scoring rule, the one-stage combination will achieve a (weakly) higher performance than the two-stage approach, asymptotically. Secondly, Part 2.\ of Theorem \ref{thm:one} shows that, at first-order, sampling variability in the performance of two-stage combinations is unaffected by sampling variability from optimizing the combination weights, and is therefore driven entirely by sampling variability from estimating the constituent models. In particular, we see in Theorem \ref{thm:one}, Part 2.\ that knowing $\eta ^{\star }$ \textit{a priori} (see the right-hand $\mathcal{S}_{0}$ term) or estimating it with $\tilde{\eta}_{n}$ (see the left-hand $\mathcal{S}_{0}$ term) makes no difference (see the subtraction) to the first-order behavior of the forecast performance measure (see the $o_{p}(n^{-1/2})$ term on the right-hand-side).
\end{remark}

\begin{remark}
Part 2.\ of Theorem \ref{thm:one} has implications for the so-called forecast combination puzzle.\footnote{The literature on the forecast combination puzzle focuses almost exclusively on linear combinations of point forecasts (\citealp{Martin2021}, is an exception). Recall from Section \ref{subsec:comb} that our results pertain to a generic combination function, and that linear combinations are indeed such a function. We conjecture that our results remain valid when the scoring rule $S$ is replaced by an appropriate reward function for point forecasts, which would imply that our results and conclusions apply in the point forecasting context.} Namely, optimizing the combination in a two-stage fashion implies that $\mathcal{S}_{0}(\tilde{\eta}_{n}, \tilde{\gamma}_{n}) \overset{p}{\rightarrow} \mathcal{S}_{0}(\eta^{\star}, \gamma^{\star}) > \mathcal{S}_{0}(\overline{\eta}, \gamma^{\star})$, for all $\overline{\eta} \in \mathcal{E} - \{\eta^{\star}\}$ (and in particular, for $\overline{\eta}$ being the vector of equal weights), where the result follows by continuity of $\mathcal{S}_{0}$, the convergence $\tilde{\theta}_{n} \overset{p}{\rightarrow} \theta^{\star}$, and since $\eta^{\star}$ maximizes $\mathcal{S}_{0}(\eta, \gamma^{\star })$. Critically, however, Theorem \ref{thm:one}, Part 2.\ demonstrates that this optimization imparts no additional variability, at first-order, into the asymptotic distribution of $\mathcal{S}_{0}(\tilde{\eta}_{n},\tilde{\gamma}_{n})$. Hence, Theorem \ref{thm:one} lends support to the optimal linear pool over its equally weighted counterpart as it pertains to forecast performance evaluations conducted via scoring rules, and in cases where $\eta ^{\star }$ is not the equally-weighted combination. However, when $\eta^{\star}$ is the equally weighted combination, Theorem \ref{thm:one}, Part 2.\ implies that estimated forecast combinations will perform very similarly to the equally weighted combination, and {that }there will be little difference between the two forecast distributions in finite-samples; at least as measured by $\mathcal{S}_{0}$ or a consistent estimator thereof.
\end{remark}

The following result compares, in our chosen measure of forecast performance, the sampling variability of one- and two-stage forecast combinations.

\begin{theorem}
\label{thm:two} If Assumptions \ref{ast:exp}-\ref{ast:dist} in \appendixname\ \ref{subsec:regularity} are satisfied and we have $\partial \mathcal{S}_0(\eta^{\star}, \gamma^{\star}) / \partial \gamma\allowbreak \ne 0$, then for $W^0$ and $W^{\star}$ as defined in \appendixname\ \ref{subsec:gmm}:

\smallskip

\noindent(One-stage) $\mathcal{S}_0(\hat{\theta}_n) = \mathcal{O}_p(n^{-1})$. More specifically, for $X\sim N(0, W^0)$, \\ $n \lbrace \mathcal{S}_0(\hat{\theta}_{n}) - \mathcal{S}_0(\theta^0) \rbrace \Rightarrow - \frac{1}{2} X^{\prime} \left[ - \partial^2 \mathcal{S}_{0}(\theta^0) / \partial \theta \partial \theta^{\prime} \right] X.$


\smallskip

\noindent(Two-stage) $\mathcal{S}_0(\tilde{\theta}_n) = \mathcal{O}_p(n^{-1/2}) $. More specifically, for $W^\star_{\gamma\gamma}$ denoting the $\gamma\gamma$-block of $W^\star$, \\ $\sqrt{n} \lbrace \mathcal{S}_0(\tilde{\theta}_{n}) - \mathcal{S}_{0}(\theta^{\star}) \rbrace \Rightarrow N(0, [\partial \mathcal{S}_{0}(\eta^{\star}, \gamma^{\star}) / \partial \gamma]^{\prime} W^{\star}_{\gamma \gamma} [\partial \mathcal{S}_0(\eta^{\star}, \gamma^{\star}) / \partial \gamma]).$


\end{theorem}

\begin{remark}
\label{rmk:rateofconvergence} Theorem \ref{thm:two} implies that not only does the one-stage approach yield \textit{better} forecast performance (see Theorem \ref{thm:one}, Part 1.), it also produces forecasts with \textit{lower sampling variability} than the two-stage combination, as measured by our forecasting accuracy measure.
\end{remark}

\begin{remark}
Part 2.\ of Theorem \ref{thm:two} demonstrates that the out-of-sample forecast performance for the two-stage combination does not depend on the sampling variability of the parameter estimates for the combination parameters $\eta $. This is consistent with the results of Theorem \ref{thm:one}, Part 2.\ in that, at first order, the sampling variability of the out-of-sample forecast performance is entirely driven by the sampling variability of the estimated constituent models.
\end{remark}

\begin{remark}
\label{rmk:derivassump} Theorem \ref{thm:two} relies on the assumption that the one- and two-stage estimators converge to different limit optimizers, that is, $\theta^{\star} \neq \theta^0$. If the limit optimizers coincide, and $\theta^{\star} = \theta^0$, Theorem \ref{thm:two} does not apply.\footnote{In view of their role as expected-score-maximizers, $\theta^{\star} = \theta^0$ if and only if $\gamma^{\star} = \gamma^0$.} For this case we require more intricate structures to systematically compare the forecast performance of the one- and two-stage estimators, and we leave this challenge for future research. For an example of a two-stage estimator that coincides with its one-stage counterpart in the limit, see \citet[sec. 2]{Pagan1986}.
\end{remark}

\section{Monte Carlo Analysis\label{sec:montecarlo}}

\subsection{DGP, Forecast Combination and Scoring Rules\label{subsec:montecarlodgp}}

Figures \ref{fig:sim} through \ref{fig:sim4} illustrate Theorems \ref{thm:one} and \ref{thm:two} in the case of a simple example. Data $Y_1, Y_2, \ldots, Y_n$ is generated from the following left- and right-censored\footnote{Censoring removes rare but extreme outliers that corrupt simulation results. By drawing a sample of $10^7$ draws from $Y_t$, we found that $\mathrm{Pr}(Y_t\ \mathrm{censored}) \approx 0.0035$.} first-order autoregressive (AR(1)) process with conditionally Gaussian errors and with a first-order autoregressive conditional heteroscedasticity (ARCH(1)) structure: 
\begin{equation*}
Y_t = \min(\max(X_t,-5),5),\ X_t = 0.5 X_{t-1} + V_t Z_t,\ V_t^2 = 0.2 + 0.75 V_{t-1}^2 Z_{t-1}^2,\ Z_t \overset{i.i.d.}{\sim} N(0,1).
\end{equation*}
The true DGP is stationary, with an unconditional mean of zero and an unconditional standard deviation of approximately 0.93.

Two constituent models are used to define the forecast combination: a Gaussian AR(1) model, 
\begin{equation}
Y_t = \alpha_0 + \alpha_1 Y_{t-1} + \sigma Z_t,\ Z_t \overset{i.i.d.}{\sim} N(0, 1), \label{eqn:fmod1}
\end{equation}
and a constant-mean model, with normally distributed innovations and an ARCH(1) structure for the conditional variance, 
\begin{equation}
Y_t = \mu + V_t Z_t,\ V^2_t = \beta_0 + \beta_1 V^2_{t-1} Z^2_{t-1},\ Z_t \overset{i.i.d.}{\sim} N(0, 1). \label{eqn:fmod2}
\end{equation}
These two models have, respectively, the following one-step-ahead predictive CDFs for $Y_t$ conditional on $Y_{1:t-1}$: 
\begin{equation*}
\textstyle
F_1(y_t \mid Y_{1:t-1} ; \gamma_1) = \Phi \left( \frac{y_t - \alpha_0 - \alpha_1 Y_{t-1}}{\sigma} \right),\ F_2(y_t \mid Y_{1:t-1} ; \gamma_2) = \Phi \left( \frac{y_t - \mu}{\sqrt{\beta_0 + \beta_1 (Y_{t-1} - \mu)^2}} \right),
\end{equation*}
where $\gamma_1 = [ \alpha_0\ \alpha_1\ \sigma^2 ]^{\prime}$, $\gamma_2 = [ \mu\ \beta_0\ \beta_1 ]^{\prime}$ and $\Phi$ is the CDF of the standard normal distribution. Our forecast combination is a linear pool of these two models, so the one-step-ahead predictive CDF of the combination is 
\begin{equation}
F_c(y_t \mid Y_{1:t-1} ; \theta) = \eta F_1(y_t \mid Y_{1:t-1} ; \gamma_1) + (1 - \eta) F_2(y_t \mid Y_{1:t-1} ; \gamma_2),  \label{eqn:predcdfsim}
\end{equation}
where $\eta \in [0,1]$ is the weight assigned to the AR(1) forecast, $1-\eta$
is the weight assigned to the ARCH(1) forecast, and $\theta = [\eta\ \gamma^{\prime}]^{\prime}$. Let $F^{\theta}_{c,t} = F_c(\, \cdot \mid Y_{1:t-1} ; \theta)$.

Since our forecast combination is supported on $\mathbb{R}$ and our DGP is supported on $[-5, 5]$, our forecast combination is clearly misspecified. Censoring occurs rarely however, so it is pertinent to discuss the similarity of our forecast combination to $X_t$, the stochastic process of the DGP before censoring. Whereas $X_t$ has a normal forecast distribution, our forecast combination is a mixture of two normal distributions, and is therefore normal if and only if $\eta = 0$ or $\eta = 1$. If $\eta = 0$, the forecast combination fails to capture the AR characteristics of $X_t$, and if $\eta = 1$, the forecast combination fails to capture the ARCH characteristics of $X_t$. Even without censoring, our forecast combination is therefore unable to exactly represent the true DGP, no matter how the parameters are estimated, nor how large is the number of observations available.

We make use of the log score and the censored log score defined in \eqref{eqn:simsl} and \eqref{eqn:simscls}, respectively. The log score is a ``local'' scoring rule, returning a high value if the realized observation is in the high-density region of the distributional forecast combination, and a low value otherwise. Our chosen censored log score, on the other hand, prioritizes accurate forecasts on the set $B = (-\infty, F^{-1}(0.2)]$ defined by the left 20\% tail of the stationary distribution of the true DGP, where $F^{-1}$ is the quantile function of this distribution. The quantile function $F^{-1}$ is estimated using $10^7$ draws from the DGP. In what follows, we abbreviate the log score as LS, and the censored log score as $\mathrm{CS}_{<20\%}$.

\subsection{Simulation Design}

To illustrate Theorems \ref{thm:one} and \ref{thm:two}, we require estimates of the expectation and variance of the forecast performance $\mathcal{S}_0(\hat{\vartheta}_n)$ in \eqref{eqn:S0}, with respect to the sampling distribution of some parameter estimator denoted generically by $\hat{\vartheta}_n$, for any given score $S$ and sample size $n$. The nature of $\hat{\vartheta}_n$ is determined by the approach used to estimate the forecast combination. We also require estimates of the expected average score of the true DGP, $\mathcal{S}_{DGP} \coloneqq \lim_{n \to \infty} \mathbb{E}[\frac{1}{n} \sum_{t=1}^n S(F_{0,t}, Y_t)]$, and of the limit optimizer, $\theta^{\star}$, where we use the notation $F_{0,t}$ to indicate the true one-step-ahead predictive based on data up to time $t-1$.

We will detail our precise simulation design shortly. For now, we make the following observations. In general, none of the required moments are available in closed form; nor is $\theta ^{\star }$. Our performance measure  $\mathcal{S}_{0}(\cdot)$ is an expectation, however, as is $\mathcal{S}_{DGP}$, so if some standard regularity conditions hold we can evaluate these to any desired degree of accuracy as the sample mean from a sufficiently long realization drawn from the true DGP. An ``error-free'' estimate of $\theta ^{\star }$ is then obtained by optimizing the estimate of $\mathcal{S}_{0}(\cdot)$.

To characterize the sampling variation of the expectation $\mathcal{S}_0(\cdot)$ as a function of the parameters $\hat{\vartheta}_n$ of the estimated forecast combination, we need estimates of the mean and variance of $\mathcal{S}_0(\hat{\vartheta}_n)$, which we obtain by using sampling in three ways. First, to produce a single $\hat{\vartheta}_n$. Second, to estimate $\mathcal{S}_0(\cdot)$ (as described above) at the given $\hat{\vartheta}_n$. Third, to replicate this process, in order to estimate the mean and variance of the performance measure $\mathcal{S}_0(\hat{\vartheta}_n)$ by the sample mean and sample variance, respectively, over $M$ replicated estimates, $\hat{\vartheta}_n$. We cannot estimate the mean and variance of $\mathcal{S}_0(\hat{\vartheta}_n)$ with a large enough value of $M$ to render the sampling error negligible, because of the need to optimize to produce $\hat{\vartheta}_n$ at each iteration. We therefore quantify this error by producing confidence intervals for the mean and variance of $\mathcal{S}_0(\hat{\vartheta}_n)$ via standard asymptotics for i.i.d.\ draws. Below, we use a pre-subscript $i$ on an estimator ${}_i \hat{\vartheta}_T$ to indicate that the corresponding forecast combination is optimized according to the score $S^i$, for some $i \in \{\mathrm{LS}, \mathrm{CS}_{<20\%}\}$ (i.e., for some $S^i \in \{S^{\mathrm{LS}}, S^{\mathrm{CS}_{<20\%}}\}$). We also use a post-superscript $j$ on a score $S^j$ used to measure performance. In this way, the indices $i$ and $j$ distinguish between the score $S^i$ used to estimate the parameters, and the score $S^j$ used to measure performance, which are not always the same.

In detail, we perform the following steps:

\medskip

\noindent1. Draw $10^6$ observations $y^{(0)}_{1:10^6}$ from the true DGP, given in the previous subsection.

\medskip

\noindent2. Produce the two-stage estimates\textbf{\ }${}_{i}\tilde{\theta}_{10^{6}}$ in \eqref{eqn:est2s2}, via the scores $i=\mathrm{LS},\mathrm{CS}_{<20\%}$, and using $y_{1:10^{6}}$. With this number of observations being very large, we view the sample criterion in \eqref{eqn:est2s1} as an error-free estimate of $\mathcal{S}_{0}(\cdot )$, and ${}_{i}\tilde{\theta}_{10^{6}}$ as thus an error-free estimate of $\theta ^{\star }$. In Step 5 we will extract the combination function parameters ${}_{i}\tilde{\eta}_{10^{6}}$ from ${}_{i}\tilde{\theta}_{10^{6}}$ as an essentially exact representation of $\eta ^{\star }$ given score $i$.

\medskip

\noindent3. Construct the true one-step-ahead predictives $F_{0,t}, t = 1, 2, \ldots, 10^6$ and compute 
\begin{equation}
\frac{1}{10^6} \sum_{t = 1}^{10^6} S^j(F_{0,t}, y^{(0)}_t) \approx \lim_{n \to \infty} \mathbb{E} \left[ \frac{1}{n} \sum_{t=1}^n S^j(F_{0,t}, Y_t) \right] = \mathcal{S}_{DGP}  \label{eqn:sdgp}
\end{equation}
for each $j \in \{\mathrm{LS}, \mathrm{CS}_{<20\%}\}$, where $S^{\mathrm{LS}} $ and $S^{\mathrm{CS}_{<20\%}}$ are the log score and censored log score, respectively, defined in the previous subsection. This quantity approximates the highest possible forecast performance $\mathcal{S}_{DGP}$ -- that attained by the true DGP -- which we use to benchmark the performance of our estimators.

\medskip

\noindent4. Draw 250000 observations $y^{(1)}_{1:250000}$ from the true DGP, independently of Step 1.

\medskip

\noindent5. Produce the one- and two-stage estimates, ${}_{i}\hat{\theta}_{n}^{(1)}$ and ${}_{i}\tilde{\theta}_{n}^{(1)}$, for the forecast combination given in the previous subsection, using $n$ of the observations drawn in Step 4 ($y_{1:n}^{(1)}$), according to the scores $i\in \{\mathrm{LS},\mathrm{CS}_{<20\%}\}$, for sample sizes $n=500,501,\ldots ,2000$. Also produce the two-stage estimator with the combination parameters fixed at their limiting values, ${}_{i}\tilde{\theta}_{n}^{(1)}|_{\eta ^{\star }} \coloneqq [ \eta ^{\star}\ {}_{i}\tilde{\gamma}_{n}^{(1)\prime} ]^{\prime}$, for the same scores $i$ and sample sizes $n$.

\medskip

\noindent6. Approximate the out-of-sample forecast performance using the average $\hat{\mathcal{S}}_{0}(\hat{\vartheta}^{(1)}_n) = \frac{1}{100n} \sum_{t = 250001 - 100n}^{250000} S^j(\hat{F}_{c,t}, y^{(1)}_{t})$ for all $\hat{\vartheta}^{(1)}_n \in \{{}_i\hat{\theta}^{(1)}_n, {}_i\tilde{\theta}^{(1)}_n, {}_i\tilde{\theta}^{(1)}_n|_{\eta^{\star}}; i = \mathrm{LS}, \mathrm{CS}_{<20\%}, n = 500,\allowbreak 501,\allowbreak \ldots,\allowbreak 2000\}$ and $j \in \{ \mathrm{LS}, \mathrm{CS}_{<20\%} \}$, where $\hat{F}_{c,t}$ is the combination's forecast distribution for $Y_t$ corresponding to the estimator $\hat{\vartheta}^{(1)}_n$ and past observations $y^{(1)}_{1:t-1}$. Given the large number of observations used to compute $\hat{\mathcal{S}_0}(\hat{\vartheta}_n^{(1)})$, this estimate is viewed as an error-free representation of $\mathcal{S}_0(\hat{\vartheta}_n^{(1)})$.

\medskip

\noindent7. Repeat Steps 4-6 1000 times, to obtain 1000 independent draws from the sampling distribution of the one-stage forecast performances $\mathcal{S}_0({}_i \hat{\theta}^{(1)}_n)$, $\mathcal{S}_0({}_i \hat{\theta}^{(2)}_n)$, \ldots, $\mathcal{S}_0({}_i \hat{\theta}^{(1000)}_n)$, the two-stage forecast performances $\mathcal{S}_0({}_i \tilde{\theta}^{(1)}_n)$, \ldots, $\mathcal{S}_0({}_i \tilde{\theta}^{(1000)}_n)$ and the two-stage forecast performances with fixed combination function parameters $\mathcal{S}_0({}_i \tilde{\theta}^{(1)}_n|_{\eta^{\star}})$, \ldots, $\mathcal{S}_0({}_i \tilde{\theta}^{(1000)}_n|_{\eta^{\star}})$, for each score $i$ used to estimate the parameters, each score $j$ used to measure performance, and each sample size $n$.

\medskip

\noindent8. Using the draws from each generic sampling distribution, approximate the first two moments by their sample counterparts: $\overline{\hat{\mathcal{S}}_0(\hat{\vartheta}_n)} = \frac{1}{1000} \sum_{m = 1}^{1000} \hat{\mathcal{S}}_0(\hat{\vartheta}^{(m)}_n) \approx \mathbb{E}[\mathcal{S}_0(\hat{\vartheta}_n)], $ and $\widehat{\mathrm{Var}}(\hat{\mathcal{S}}_0(\hat{\vartheta}_n)) = \frac{1}{1000} \sum_{m = 1}^{1000} \left( \hat{\mathcal{S}}_0(\hat{\vartheta}^{(m)}_n) - \overline{\hat{\mathcal{S}}_0(\hat{\vartheta}_n)} \right)^2 \approx \mathrm{Var}(\mathcal{S}_0(\hat{\vartheta}_n)),$ for all $\hat{\vartheta}_n \in \{{}_i\hat{\theta}_n, {}_i\tilde{\theta}_n, {}_i\tilde{\theta}_n|_{\eta^{\star}} ;\allowbreak i = \mathrm{LS}, \mathrm{CS}_{<20\%}, n = 500, 501, \ldots, 2000\}$ where $\mathcal{S}_0$ is defined with respect to $S^j, j \in \{\mathrm{LS}, \mathrm{CS}_{<20\%}\}$.

\medskip

\noindent9. Calculate the 95\% confidence intervals for the population moments $\mathbb{E}[\mathcal{S}_0(\hat{\vartheta}_n)]$ and $\mathrm{Var}(\mathcal{S}_0(\hat{\vartheta}_n))$ using the sample moments in Step 8 and standard asymptotics for i.i.d.\ draws.

\medskip

\noindent10. Using the approximation for $\mathcal{S}_{DGP}$ in \eqref{eqn:sdgp} and the sample moments from Step 8, produce point estimates and confidence intervals for the expected divergence $\mathcal{S}_{DGP} - \mathbb{E}[\mathcal{S}_0(\hat{\vartheta}_n)]$. Note that the score $j$ used to measure performance is identical in the definition of $\mathcal{S}_{DGP}$ and $\mathbb{E}[\mathcal{S}_0(\hat{\vartheta}_n)]$.

\subsection{Results}

Figures \ref{fig:sim} and \ref{fig:sim2} plot the expected divergence, $\mathcal{S}_{DGP} - \mathbb{E}[\mathcal{S}_0(\hat{\vartheta}_n)]$, in log score and censored log score terms, for a forecast combination parameterized by $\hat{\vartheta}_n$ and produced according to either the log score or the censored log score, and in either a one- or two-stage fashion. Approximation of this quantity, and construction of the corresponding 95\% confidence intervals, occurs as in Steps 8-10 of the simulation instructions above. A lower expected divergence, and thereby a lower value on the vertical axis of any of these plots, indicates a better expected forecast performance, for a given sample size $n$. That is, on average, $\mathcal{S}_0(\hat{\vartheta}_n)$ is closer to the maximum possible value, $\mathcal{S}_{DGP}$. For the moment, we will restrict our attention to the results for the one-stage (red) and two-stage (green) estimators, returning to discuss the two-stage estimator with the weight, $\eta$ (given in \eqref{eqn:predcdfsim}), fixed at its limit optimizer (blue) at the end of the section.

\begin{figure*}[t]
\includegraphics[width=\textwidth]{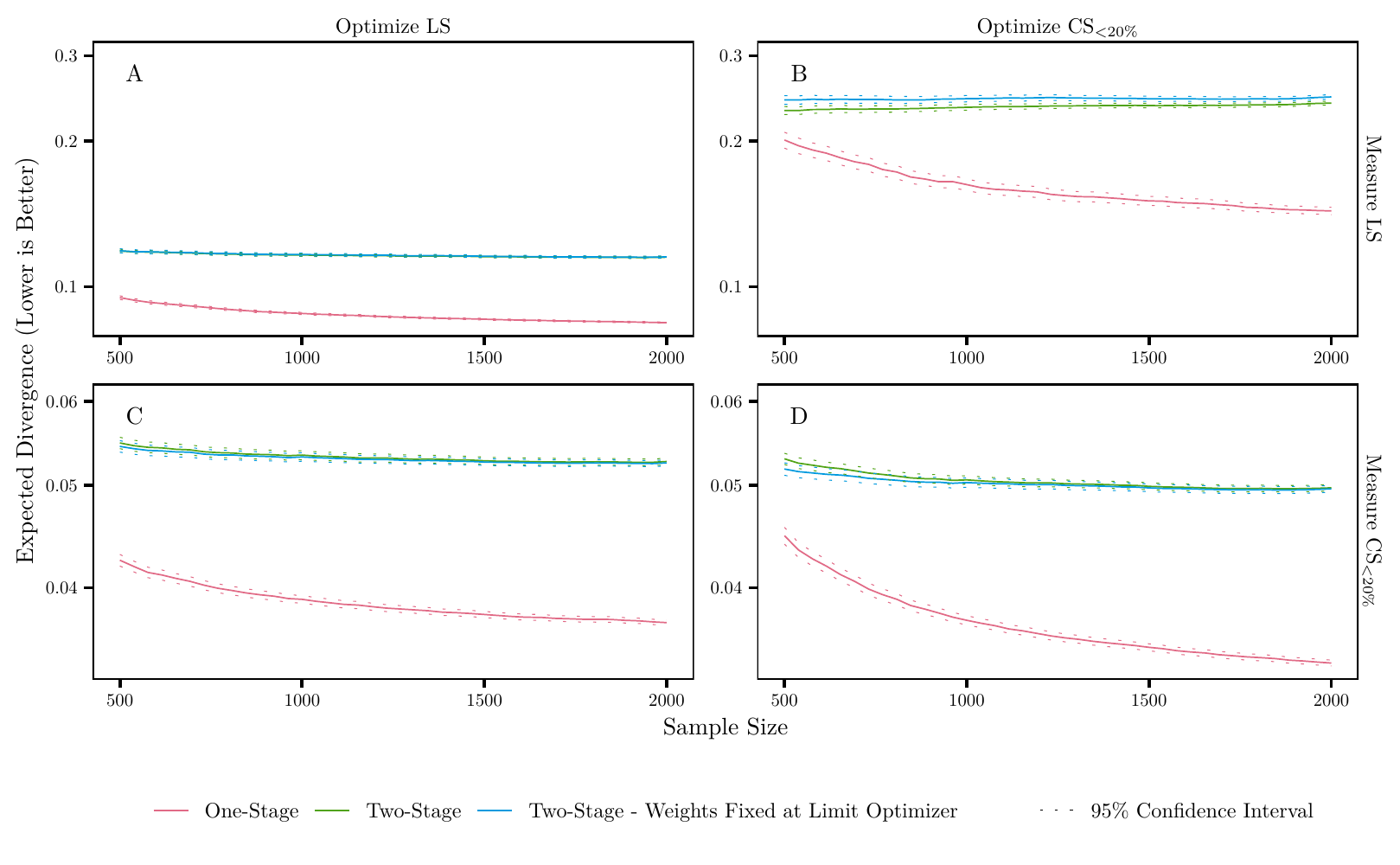}
\caption{The expectation of the difference between the out-of-sample one-step-ahead forecast performance of a misspecified forecast combination and that of the true DGP, over a range of sample sizes. \cite{Gneiting2007} call this quantity the expected divergence. The forecast combination is optimized in a one-stage fashion (red), a two-stage fashion (green) or in a way that comprises the first stage of the two-stage combination, followed by a fixed combination at $\tilde{\eta}_n = \eta^{\star}$ (blue). Parameters are optimized according to the log score (A and C, first column) or a censored log score that prioritizes accuracy in the lower 20\% tail of the forecast distribution (B and D, second column). The divergence is measured on the vertical axes according to the log score (A and B, first row) or the censored log score (C and D, second row). The expectations and confidence intervals are constructed as per Steps 8-10 in the text, with the 95\% confidence bounds appearing as small dashed lines.}
\label{fig:sim}
\end{figure*}

\begin{figure*}[t]
\includegraphics[width=\textwidth]{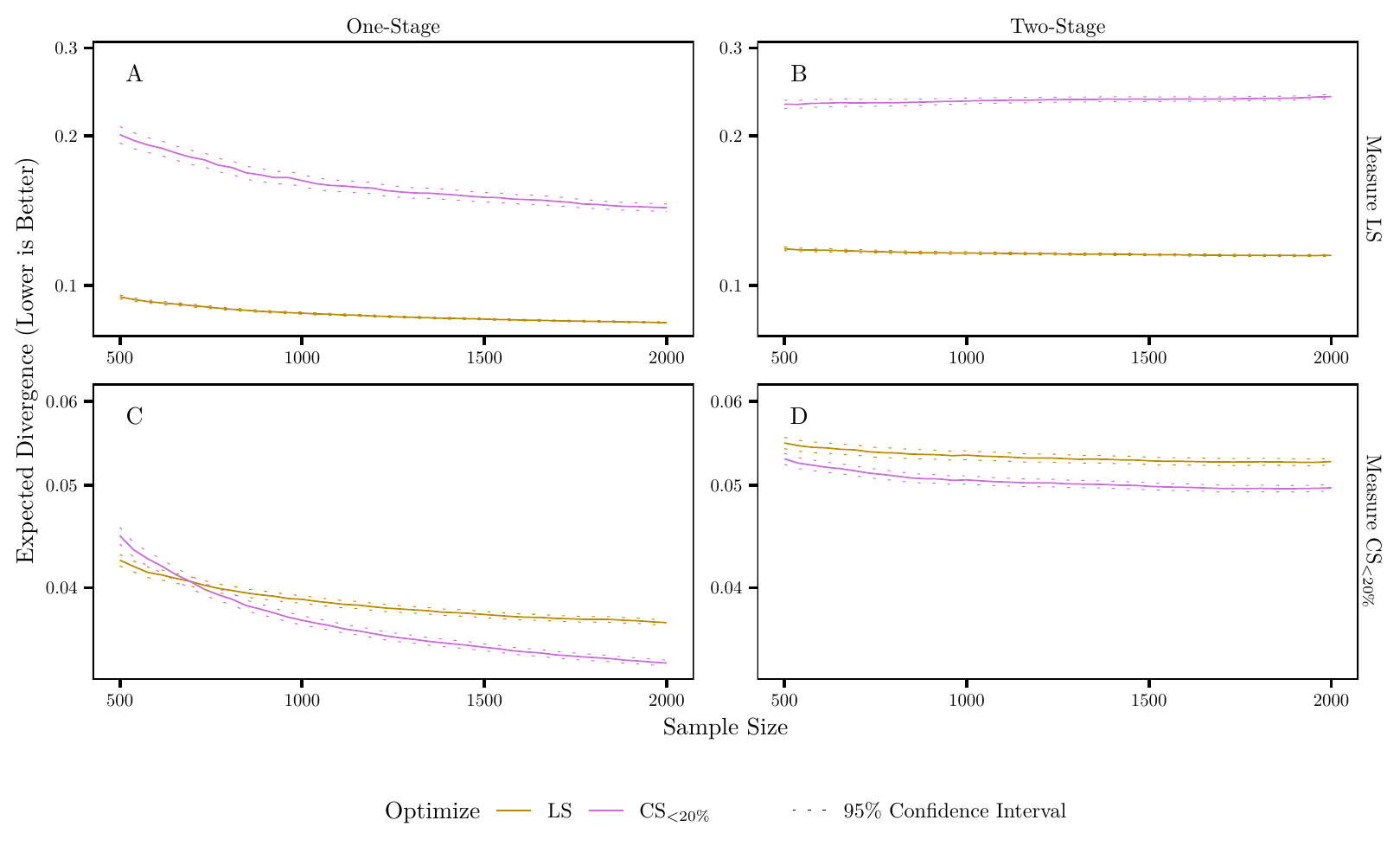}
\caption{The expectation of the difference between the out-of-sample one-step-ahead forecast performance of a misspecified forecast combination and that of the true DGP, over a range of sample sizes. \cite{Gneiting2007} call this quantity the expected divergence. The forecast combination is optimized in a one-stage (A and C, first column) or two-stage (B and D, second column) manner according to the log score (gold) or a censored log score that prioritizes accuracy in the lower 20\% tail of the forecast distribution (purple). The divergence is measured on the vertical axes according to the log score (A and B, first row) or the censored log score (C and D, second row). The expectations and confidence intervals are constructed as per Steps 8-10 in the text, with the 95\% confidence bounds appearing as small dashed lines.}
\label{fig:sim2}
\end{figure*}

\begin{figure*}[t]
\includegraphics[width=\textwidth]{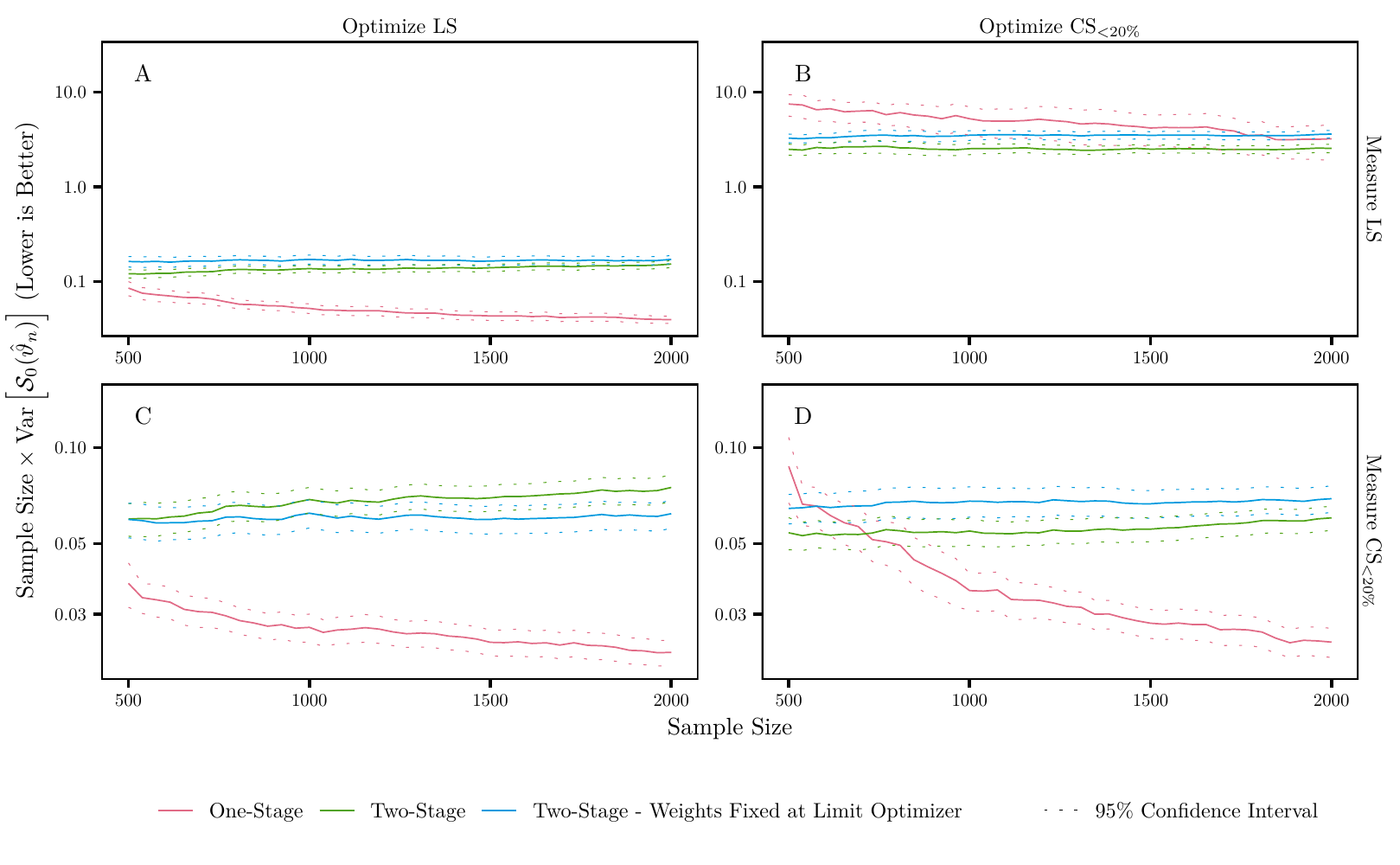}
\caption{The sample size multiplied by the variance of the one-step-ahead forecast performance $\mathcal{S}_0(\hat{\vartheta}_n)$ of a misspecified forecast combination, over a range of sample sizes. Here, the forecast combination parameter estimator $\hat{\vartheta}_n$ optimizes the average log score (A and C, first column) or censored log score (B and D, second column) in a one-stage fashion (red), a two-stage fashion (green) or in a way that comprises the first stage of the two-stage combination, followed by a fixed combination at $\tilde{\eta}_n = \eta^{\star}$, the optimal two-stage combination function parameter values (blue). The censored log score prioritizes accuracy in the lower 20\% tail of the forecast distribution, and the variance of the forecast performance is measured on the vertical axes according to the log score (A and B, first row) or the censored log score (C and D, second row). The variances and confidence intervals are constructed as per Steps 8-10 in the text, with the 95\% confidence bounds appearing as small dashed lines.}
\label{fig:sim3}
\end{figure*}

\begin{figure*}[t]
\includegraphics[width=\textwidth]{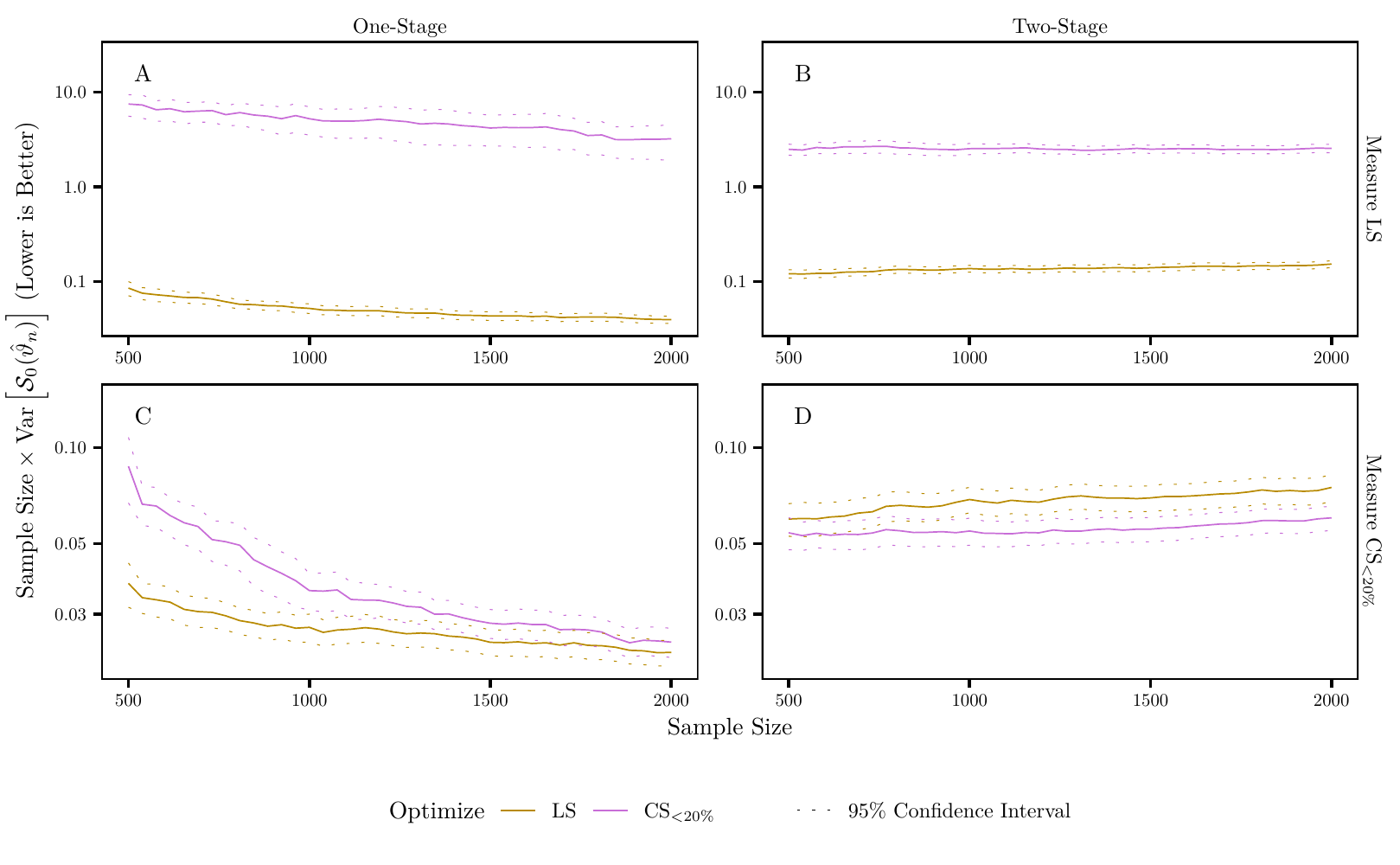}
\caption{The sample size multiplied by the variance of the one-step-ahead forecast performance $\mathcal{S}_0(\hat{\vartheta}_n)$ of a misspecified forecast combination, over a range of sample sizes. Here, the forecast combination parameter estimator $\hat{\vartheta}_n$ is optimized in a one-stage (A and C, first column) or two-stage (B and D, second column) fashion according to the log score (gold) or a censored log score that prioritizes accuracy in the lower 20\% tail of the forecast distribution (purple). The variance of the forecast performance is measured on the vertical axes according to the log score (A and B, first row) or the censored log score (C and D, second row). The variances and confidence intervals are constructed as per Steps 8-10 in the text, with the 95\% confidence bounds appearing as small dashed lines.}
\label{fig:sim4}
\end{figure*}

In the diagonal Panels A and D of Figure \ref{fig:sim}, performance is measured and forecast combinations are produced according to the same score. In these plots, we see that the one-stage approach (red) has better expected forecast performance (a lower value on the vertical axis) than the two-stage approach (green) when measuring forecast performance and producing forecast combinations according to the same score. This is consistent with the Theorem \ref{thm:one}, Part 1.\ result that the limiting forecast performance of the one-stage forecast combination is greater than that of the two-stage forecast combination when the same score is used to measure performance and produce forecast combinations.

In the off-diagonal Panels B and C, the scores used to measure performance and produce forecast combinations do not coincide and Theorem \ref{thm:one} does not apply, but uniform dominance of the one-stage forecast combination over the two-stage alternative nevertheless persists.

Figure \ref{fig:sim2} displays the same results as in Figure \ref{fig:sim}, but through a different lens, highlighting the advantage for either method (one-stage or two-stage) of producing forecast combinations according to the score used to measure forecast performance. Since our forecast combination is misspecified, optimizing for different scores (via either the one- or two-stage approach) leads, in general, to parameter estimates with different limit optimizers, and therefore to forecast combinations with different limiting expected average scores. Panels A and B show that producing forecast combinations according to the log score (gold) leads to a better expected forecast performance in terms of log score (a lower value on the vertical axis), across all sample sizes, relative to optimizing according to the censored log score (purple). The extent of the dominance is similar for both the one- and two-stage approaches. Likewise, Panels C (for sample sizes over 700) and D (for all sample sizes) show that producing forecast combinations according to the censored log score (purple) leads to a better expected forecast performance in terms of censored log score (a lower value on the vertical axis), relative to optimizing according to the log score (gold); and for both the one- and two-stage methods.

Figures \ref{fig:sim3} and \ref{fig:sim4} show the variance of the forecast performance of the combination, in log score and censored log score terms, for forecast combinations that optimize either the log score or the censored log score. Once again, all computational details are given in Steps 8-10 of the simulation instructions above. A lower value on the vertical axis indicates a forecast combination whose forecast performance has a lower sampling variability. The variance has been multiplied by the sample size to offset the $n^{-1/2}$ convergence of most estimators, the exception being the one-stage combination where the optimization criterion and the forecast performance measure coincide. This implies, of course, that the figures for this exceptional case (red in Panels A and D of Figure \ref{fig:sim3}, and gold and purple in Panels A and C, respectively, of Figure \ref{fig:sim4}) ought to decline towards $-\infty$ as the sample size increases, since the forecast performance of this forecast combination converges as $n^{-1}$ (see Theorem \ref{thm:two}, Part 1.\ and Remark \ref{rmk:rateofconvergence}). This is indeed visible in all cases, despite the slight plateauing effect arising due to Monte Carlo error. Once again, we defer discussion of the two-stage estimator with the weight fixed at its limit optimizer (blue) until the end of the section.

In Panels A and D of Figure \ref{fig:sim3}, in which the optimization criterion and the forecast performance measure coincide, we see that the one-stage approach (red) has a forecast performance with a lower sampling variability for large sample sizes than the two-stage approach (green). This is consistent with the Theorem \ref{thm:two} result that the forecast performance of the one-stage forecast combination converges faster than that of the two-stage forecast combination when the same score is used to measure performance and optimize the parameters (see Remark \ref{rmk:rateofconvergence}). In the off-diagonal Panels B and C, the optimization criterion and the forecast performance measure do not coincide. The uniform dominance of the one-stage forecast combination remains in evidence in Panel C. However, in Panel B, there is no such dominance on view.

Figure \ref{fig:sim4} mimics Figure \ref{fig:sim2}, but with the (scaled) variance of the score being the focus. The advantage for both methods of producing forecast combinations according to the score used to measure forecast performance is on display, with Panel C illustrating the lone exception. In Panels A and B, we see that producing the forecast combination according to the log score (gold) leads to a smaller sampling variability in log-score forecast performance (a lower value on the vertical axis), across all sample sizes, relative to optimizing according to the censored log score (purple); and for both types of combinations. Panel D shows that, for the two-stage combination and for sample sizes over 1000 (at which point the two sets of confidence intervals do not overlap), optimizing according to the censored log score (purple) leads to a smaller sampling variability in censored-log-score forecast performance, relative to optimizing according to the log score (gold). For the one-stage combination in Panel C however, the sampling variability of the censored-log-score forecast performance is smaller for the combination optimizing the log score (gold) than for the combination optimizing the censored log score (purple), so long as sample sizes are small. For large sample sizes, optimizing either score leads to a similar censored-log-score sampling variability (with overlapping confidence intervals).

Finally, with reference to Panels A and D in both Figures \ref{fig:sim} and \ref{fig:sim3} -- in which the scores used to produce forecast combinations and measure performance coincide -- we see that both the (expected) forecast performance and the forecast performance variance are indistinguishable (that is, have overlapping confidence intervals), or at least are very similar, for the two-stage forecast combination (green) and the two-stage forecast combination with the combination function parameter (weight) fixed at its limit optimizer (blue). This is consistent with Theorem \ref{thm:one}, Part 2., which holds that the performance of these two forecast combinations converge to the same asymptotic distribution (given in Theorem \ref{thm:two}, Part 2.) if the score of the optimization criterion coincides with the score of the forecast performance measure. In this context, the limiting sampling variability of the forecast performance of the two-stage forecast combination derives entirely from sampling variability in the estimation of the constituent models, and the numerical results are simply highlighting this fact.

On the off-diagonal panels of these two figures (Panels B and C), where performance measurement and combination production are conducted according to different scores, results are mixed. Confidence intervals around the green and blue lines overlap in Panel C of Figures \ref{fig:sim} and \ref{fig:sim3}. For Panel B in Figures \ref{fig:sim} and \ref{fig:sim3} on the other hand, confidence intervals around the green and blue lines do not overlap, indicating that both the (expected) forecast performance and the forecast performance variance are different between the two two-stage forecast combinations when the combinations are estimated according to the censored log score and performance is measured according to the log score, at least for the sample sizes considered.

The results of this simulation exercise have been obtained via repeatedly sampling from a known true DGP. In empirical settings, the true DGP is unknown, and we must rely only on the observations before us. In the next section, we illustrate how to approximate the sampling distribution of the forecast performance, without the expectation of correct specification. The results illustrated in the above simulation exercise will be shown to be robust to this added complexity.

\section{Performance on S\&P500 Returns\label{sec:emp}}

In this empirical exercise, we explore the extent to which selected theoretical results of Section \ref{sec:implications} are reflected in estimated forecast combinations for S\&P500 returns.

\subsection{Data, Forecast Combination and Scoring Rules\label{subsec:empdatacomb}}

Our dataset contains 8565 daily continuously-compounded returns $y_{1},\ldots ,y_{8565}$ extending from January 5th, 1988 to December 31st, 2021, and was obtained from the Global Financial Data database. The constituent forecasting models are as given in \eqref{eqn:fmod1} and \eqref{eqn:fmod2} in Section \ref{subsec:montecarlodgp}, which are combined via the linear pool to give the combination specified by \eqref{eqn:predcdfsim}. Whilst these constituent models are admittedly less sophisticated than those that would typically be used to model financial returns, and the simple linear combination of such models less ambitious than other combination approaches adopted in the returns literature (e.g.\ \citealp{Billio2013}), these choices are sufficient for the purpose at hand. That purpose is simply to illustrate the role played by the different forms of sampling variability in the production of estimated forecast combinations; raw forecast accuracy, and the attainment of that accuracy via judicious model selection and combination, is not our goal.

Forecast combinations are produced in either a one- or two-stage fashion, and via optimization of a criterion function based on one of three scores: the log score (LS), and the censored log scores focusing on the lower 10\% ($\mathrm{CS}_{<10\%}$) and 20\% ($\mathrm{CS}_{<20\%}$) tails. The score $\mathrm{CS}_{<20\%}$ is defined at the end of Section \ref{subsec:montecarlodgp}, as is $\mathrm{CS}_{<10\%}$ on replacing ``$0.2$'' by ``$0.1$''. We have selected the log score because it is ubiquitous, and chosen\textbf{\ }censored log scores with a focus on the lower tails for their relevance in financial settings where a primary goal is the accurate prediction of large or unlikely losses. To measure predictive performance, the average out-of-sample LS, $\mathrm{CS}_{<10\%}$ and $\mathrm{CS}_{<20\%}$ values are calculated for two distinct evaluation periods: the ``overall'' period and an ``extreme'' period. The overall period extends from January 3rd, 2017 to December 31st, 2021, and includes the first two years of the COVID-19 pandemic and the three years prior. For the extreme period we consider returns from the first six months of the pandemic, from January 2nd, 2020 to June 30th, 2020, during which the S\&P500 experienced unusually extreme returns and high volatility. For both evaluation periods, forecast combinations are produced using all returns observed before the beginning of the period according to the steps outlined in Section \ref{subsec:produce}. For forecast combinations evaluated over the overall period, the in-sample dataset has a size of $n=7306$, and the size of the evaluation period itself is $\tau = 1259$. The corresponding sample sizes are $n = 8060$ and $\tau = 125$ when evaluating performance over the extreme period.

For any forecast combination parameterized generically by the vector $\theta$, the expected average score $\mathcal{S}_{0}(\theta)$ is estimated according by $\hat{\mathcal{S}}_{0}^{j}(\theta) = \frac{1}{\tau} \sum_{t = n+1}^{n + \tau} S^{j}(F_{c,t}^{\theta}, y_{t})$, where $j$ denotes the score evaluated (one of LS, $\mathrm{CS}_{<10\%}$, or $\mathrm{CS}_{<20\%}$), $y_{t}$ is the return on day $t$ and $F_{c,t}^{\theta}$ is the corresponding predictive CDF of our forecast combination, defined immediately after \eqref{eqn:predcdfsim} in Section \ref{subsec:montecarlodgp}. Unlike the simulation exercise of Section \ref{sec:montecarlo}, we see that the expected average score must be estimated by an out-of-sample average score based on a \textit{limited} sample size (either $\tau =1259$ or $\tau =125$). Our estimate of the expected average score is therefore subject to three sources of sampling variability: 1) from in-sample estimation of the constituent models, 2) from in-sample estimation of the weight, $\eta $, and 3) from out-of-sample estimation of the expected average score via $\hat{\mathcal{S}}_{0}(\cdot)$. Throughout this exercise, we take the out-of-sample average score as given, and quantify and visualize sampling variability from parameter estimation only, leaving an accounting of the out-of-sample sampling variability as a topic for future research.

In detail, we perform the following steps:

\smallskip

\noindent1. Given returns $y_{1:n}$, produce the one- and two-stage forecast combinations parameterized by ${}_i \hat{\theta}_n$ and ${}_i \tilde{\theta}_n$, respectively, for the functional form given in \eqref{eqn:fmod1}-\eqref{eqn:predcdfsim}, according to the scores $i \in \{\mathrm{LS}, \mathrm{CS}_{<10\%}, \mathrm{CS}_{<20\%}\}$, and for sample size $n = 7306$ (for producing one-step-ahead forecasts for the overall period) and $n = 8060 $ (when forecasting the extreme period). For both the overall period and the extreme period we use a fixed estimation window, such that the same parameter estimates (conditional on the same $n$ observations) are used to produce all $\tau$ out-of-sample forecasts. Let ${}_i \theta^0$ and ${}_i \theta^{\star}$ be the weak limits of the one- and two-stage forecast combinations parameterized by ${}_i \hat{\theta}_n$ and ${}_i \tilde{\theta}_n$, respectively, as $n \to \infty$.

\smallskip

\noindent2. For all $\hat{\vartheta}_n \in \{ {}_i \hat{\theta}_n, {}_i \tilde{\theta}_n ; i = \mathrm{LS}, \mathrm{CS}_{<10\%}, \mathrm{CS}_{<20\%}, (n, \tau) = (7306, 1259), (8060, 125)\}$:

\smallskip

\noindent (a) Let $\vartheta = \mathrm{plim}_n \hat{\vartheta}_n$ (either $\vartheta = {}_i \theta^0$ or $\vartheta = {}_i \theta^{\star}$) and obtain a consistent estimate $\hat{W}_n$ of the asymptotic covariance matrix $W$ defined by $\sqrt{n} (\hat{\vartheta}_n - \vartheta) \Rightarrow N(0, W)$ as $n \to \infty$.\footnote{See \appendixname\ \ref{subsec:gmm} for a GMM representation of $\hat{\vartheta}_n$, to which many works on GMM can be applied to produce a consistent $\hat{W}_n$. We use \cite{Andrews2002}.}

\smallskip

\noindent (b) Simulate i.i.d.\ draws $\hat{\vartheta}_n^{(1)}, \hat{\vartheta}_n^{(2)}, \ldots, \hat{\vartheta}_n^{(20000)}$ from the distribution $N(\hat{\vartheta}_n, \hat{W}_n/n)$ (conditional on $\hat{\vartheta}_n$ and $\hat{W}_n/n$), and compute $\hat{\mathcal{S}}^j_0(\hat{\vartheta}_n^{(1)}), \hat{\mathcal{S}}^j_0(\hat{\vartheta}_n^{(2)}), \ldots, \hat{\mathcal{S}}^j_0(\hat{\vartheta}_n^{(20000)})$ for $j \in \{LS,\allowbreak \mathrm{CS}_{<10\%},\allowbreak \mathrm{CS}_{<20\%}\}$.

\smallskip

\noindent (c) Produce a kernel density estimate based on the sample $\hat{\mathcal{S}}_{0}^{j}(\hat{\vartheta}_{n}^{(1)}), \hat{\mathcal{S}}_{0}^{j}(\hat{\vartheta}_{n}^{(2)}), \ldots, \hat{\mathcal{S}}_{0}^{j}(\hat{\vartheta}_{n}^{(20000)})$ and make note of the empirical estimate $\hat{\mathcal{S}}_{0}^{j}(\hat{\vartheta}_{n})$ and the 95\% confidence interval $\Big(\hat{\mathcal{S}}_{0}^{j}(\hat{\vartheta}_{n}^{(500)}), \allowbreak\hat{\mathcal{S}}_{0}^{j}(\hat{\vartheta}_{n}^{(19500)})\Big)$ for the out-of-sample score $\hat{\mathcal{S}}_{0}^{j}(\vartheta )$ of the limit optimizer $\vartheta$.\footnote{We can show that the draws $\hat{\mathcal{S}}_{0}^{j}(\hat{\vartheta}_{n}^{(1)}),\hat{\mathcal{S}}_{0}^{j}(\hat{\vartheta}_{n}^{(2)}), \ldots, \hat{\mathcal{S}}_{0}^{j}(\hat{\vartheta}_{n}^{(20000)})$ are from a consistent bootstrap distribution for $\hat{\mathcal{S}}_{0}^{j}(\hat{\vartheta}_{n})$, conditional on the out-of-sample draws $y_{n+1:n+\tau}$ that define $\hat{S}_{0}^{j}$, by applying Theorem 23.5 of \cite{VanDerVaart1998} with $\phi = \hat{\mathcal{S}}_{0}^{j}$, $\theta = \vartheta$, $\hat{\theta}_{n} = \hat{\vartheta}_{n}$, $T = N(0,W)$ and $\hat{\theta}_{n}^{\ast} \overset{i.i.d.}{\sim} N(\hat{\vartheta}_{n}, \hat{W}_{n}/n)$. Our confidence intervals are produced by Efron's percentile method, described in \citet[sec. 23.1]{VanDerVaart1998}.}

\subsection{Results}

Table \ref{tbl:emp1} contains the point estimates (labeled ``Average'') and confidence intervals (labeled ``95\% CI'') for the average out-of-sample LS, $\mathrm{CS}_{<10\%}$ and $\mathrm{CS}_{<20\%}$ (columns) for the forecast combination described above optimized according to LS, $\mathrm{CS}_{<10\%}$ or $\mathrm{CS}_{<20\%}$ in either a one- or two-stage fashion (rows). Bolded figures in the ``Average'' rows indicate the highest performing method according to the out-of-sample measure in the column heading. Likewise, bolded figures in the ``95\% CI'' rows indicate the method with the narrowest confidence interval around the out-of-sample score denoted in the column heading. Recall that the hold-out sample $y_{n+1:n+\tau }$ is taken as given when interpreting the confidence intervals, which thus reflect sampling variability in forecast combination production only. In short, the theoretical results in Section \ref{sec:implications} indicate that the one-stage combination where the in-sample and out-of-sample scores coincide (appearing on the Table's diagonals) ought to have both the highest average out-of-sample score and the narrowest 95\% CI. We will discuss results for the overall period (Panel A) first, and results for the extreme period (Panel B) second. We supplement these tabulated results with a visualization of the overall (extreme) results in Figure \ref{fig:emp2} (\ref{fig:emp3}).

\begin{table}[ht]
\centering
\resizebox{\linewidth}{!}{\input{figure/TBL1.tex}}
\caption{Point estimates and confidence intervals for the average out-of-sample scores (columns) for one- and two-stage forecast combinations of S\&P500 returns that optimize a variety of average in-sample scores (rows). }
\label{tbl:emp1}
\end{table}

Panel A of the table contains results for the out-of-sample period extending from January 3rd, 2017 to December 31st, 2021. When measuring forecast performance in LS terms (first column), we find that our theoretical results are reflected in the dominance of the combination optimized by the LS in a one-stage fashion; it having both the highest average out-of-sample value (first row, bold) and the narrowest confidence interval (second row, bold). Also note that the average forecast performance of the LS-based one-stage combination is above the upper bounds of the confidence intervals for the other forecast combinations along that column, lending further support to its dominance. In the second column, where the out-of-sample score is $\mathrm{CS}_{<10\%}$, we find that the LS-based one-stage combination once again has the highest average out-of-sample score (first row, bold) with the narrowest confidence interval (second row, bold), contrary to theory, which would have the $\mathrm{CS}_{<10\%}$-based one-stage combination claim those titles. Note, however, that the $\mathrm{CS}_{<10\%}$ performance of all three one-stage forecast combinations are similar, with their average out-of-sample scores lying in each other's confidence intervals. Results for the out-of-sample $\mathrm{CS}_{<20\%}$ are displayed in the third column, where the one-stage forecast combination also optimized according to the $\mathrm{CS}_{<20\%}$ has neither the best performance (ninth row) nor the narrowest confidence interval (tenth row). Like the results for $\mathrm{CS}_{<10\%}$, we see that the $\mathrm{CS}_{<20\%}$ out-of-sample scores are also similar for all one-stage combinations, with their average out-of-sample scores lying in each other's confidence intervals.

We now focus our attention on Panel B, which contains the results for the extreme period extending from January 2nd, 2020 to June 30th, 2020. Here, we find that the one-stage combination optimizing the $\mathrm{CS}_{<10\%}$ (fifth row, bold) has the highest average out-of-sample LS (first column), which also lies outside the confidence intervals along that column for the out-of-sample LS of all other combinations. A similar observation can be made regarding the out-of-sample $\mathrm{CS}_{<10\%}$ (column two) and $\mathrm{CS}_{<20\%}$ (column three). Forecast performance in terms of these out-of-sample scores is dominated by the one-stage combination optimizing the $\mathrm{CS}_{<20\%}$ (ninth row, bold), whose average out-of-sample score lies outside the confidence intervals of all other forecast combinations for columns two and three, with the exception of the confidence intervals for the one-stage combination optimizing the $\mathrm{CS}_{<10\%}$.

\begin{figure}[t]
\includegraphics[width=\textwidth]{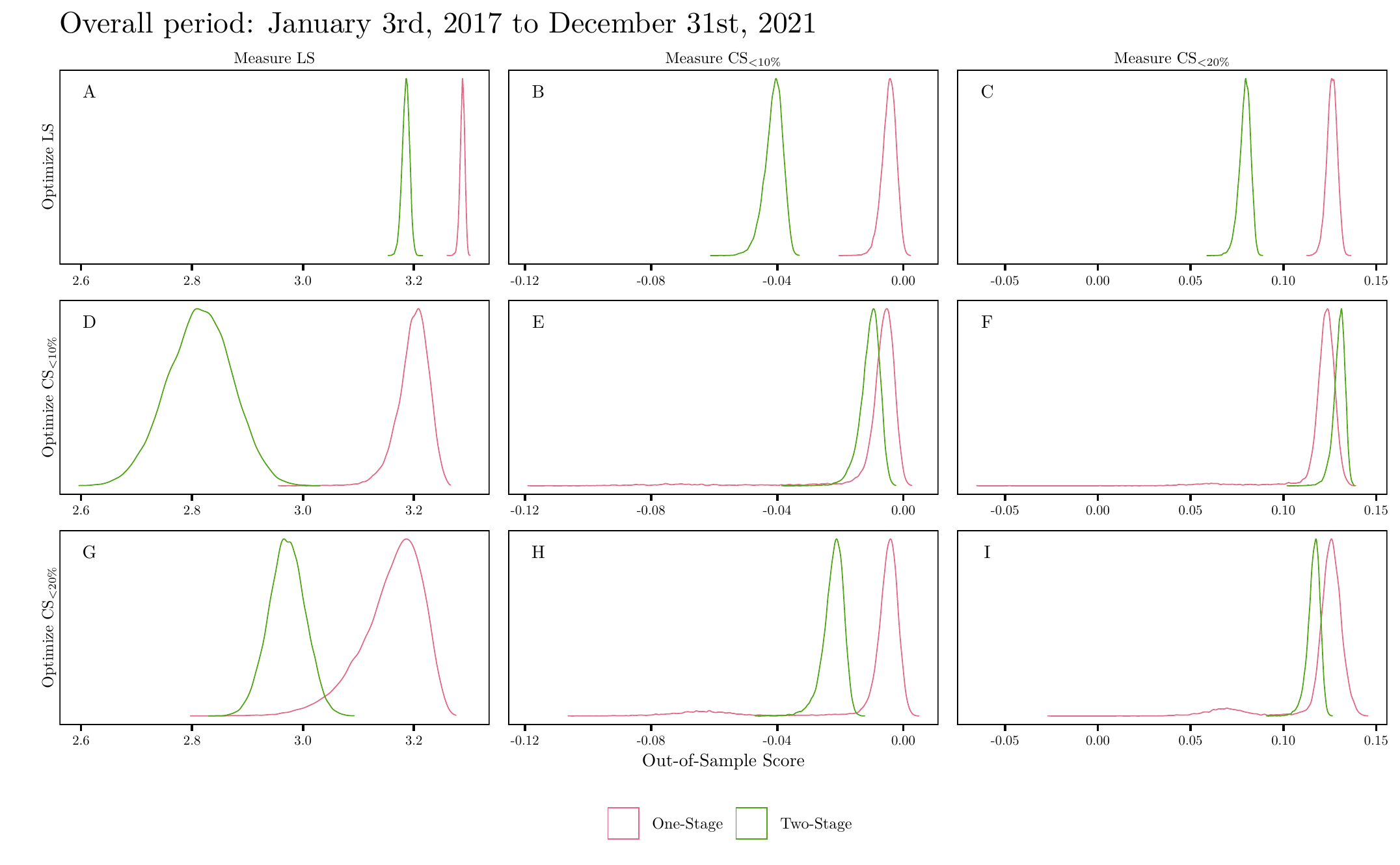}
\caption{Sampling distributions capturing the variability of average overall-period scores (columns) due to the in-sample one-stage estimation (red) and two-stage estimation (green) of forecast combinations of S\&P500 returns according to a variety of in-sample scores (rows).}
\label{fig:emp2}
\end{figure}

\begin{figure}[t]
\includegraphics[width=\textwidth]{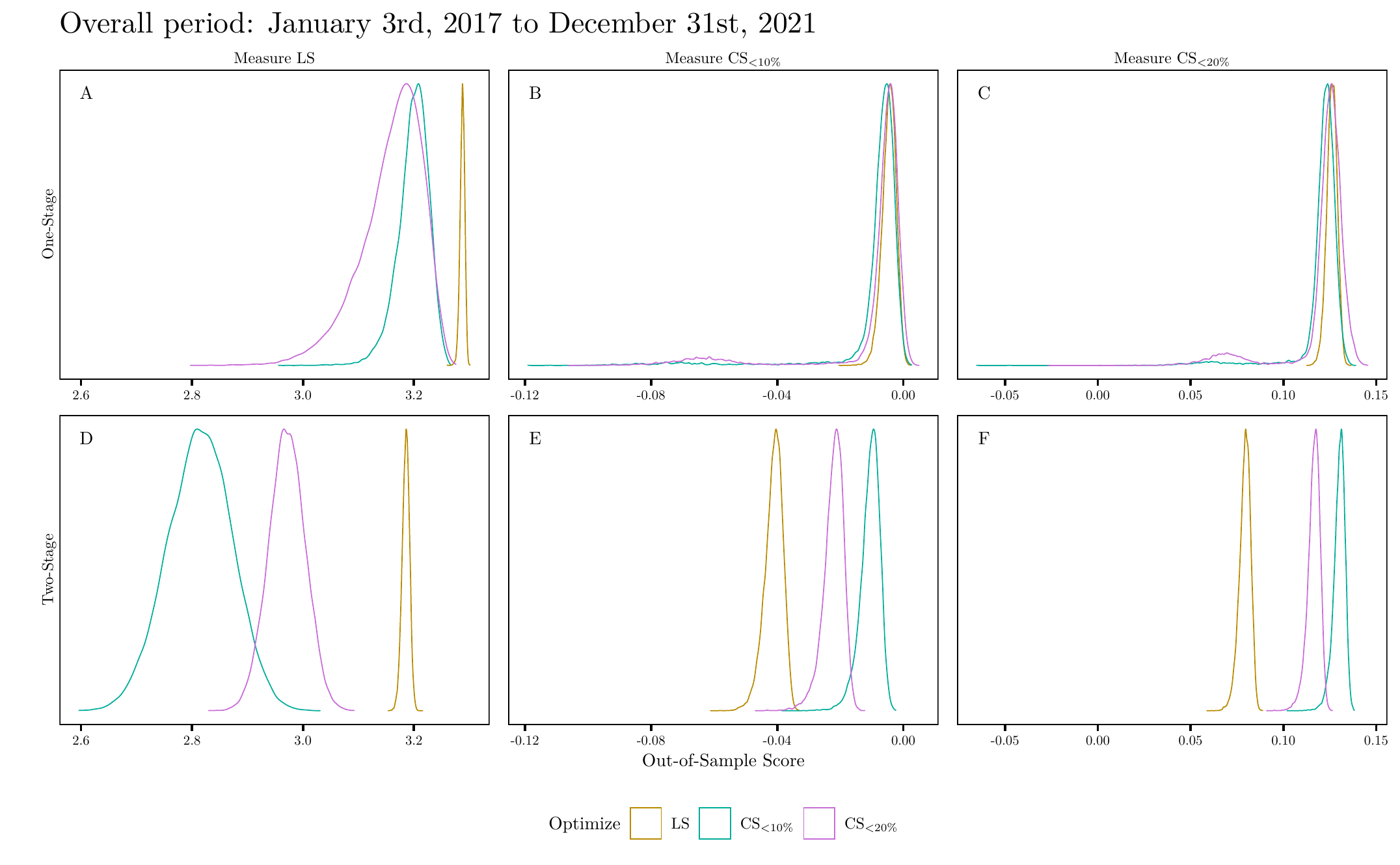}
\caption{Sampling distributions capturing the variability of average overall-period scores (columns) due to the in-sample one-stage estimation (first row) and two-stage estimation (second row) of forecast combinations of S\&P500 returns according to a variety of in-sample scores (colors).}
\label{fig:emp3}
\end{figure}

In Figures \ref{fig:emp2} to \ref{fig:emp5}, we seek to shed a different light on the results in Table \ref{tbl:emp1}, with particular attention given to differences between results based on one- and two-stage estimation. Figure \ref{fig:emp2} displays, for the overall period, the kernel density estimates for the average out-of-sample score (columns), for one-stage and two-stage forecast combinations (colors) optimized according to the LS, $\mathrm{CS}_{<10\%}$, and $\mathrm{CS}_{<20\%}$ (rows). In the diagonal Panels A, E and I, where the in-sample and out-of-sample scores are the same, the sampling distributions for the out-of-sample forecast performance of the one-stage combination (red) is higher than (that is, to the right of) the sampling distribution for the two-stage combination (green), reflecting Theorem \ref{thm:one}, Part 1. For all other panels with the exception of Panel F, the one-stage combination also has a higher forecast performance than its two-stage counterpart, despite the mismatch between the estimation and measurement criterion.

Figure \ref{fig:emp3} displays the same results as in Figure \ref{fig:emp2}, but rearranged, for the purpose of addressing the extent to which our theory is reflected in a comparison between forecast combinations produced by different scores. The rearrangement is such that rows now delineate between one- and two-stage combination, and colors indicate the score used to produce the forecast combination. Theorem \ref{thm:one} is illustrated in Panel A, where we see that optimizing according to the LS (gold) leads to a one-stage combination with an out-of-sample LS being higher (to the right of), and with less sampling variability (a narrower kernel density), relative to optimizing either of the censored log scores (teal and purple). The same is observed for two-stage estimation in Panel D, only now the $\mathrm{CS}_{<20\%}$-based estimator outperforms its $\mathrm{CS}_{<10\%}$-based counterpart. On the other hand, sampling distributions for the out-of-sample scores of the one-stage forecast combinations in Panels B and C are similar in both location and spread, with substantial overlap. Looking again at the bottom row, the two-stage combination that optimizes the $\mathrm{CS}_{<10\%}$ dominates out-of-sample forecast performance in terms of both the $\mathrm{CS}_{<10\%}$ (Panel E) and the $\mathrm{CS}_{<20\%}$ (Panel F), when compared to other two-stage combinations. Also note that the performances of combinations optimized according to the three different scores are more distinct for two-stage combinations in the bottom row, than for the one-stage combinations on the top row. Further, with the exception of the small LS sampling variability for the LS-optimizing combination in Panels A and D of the first column, the sampling variability of the out-of-sample score is similar (similar kernel density widths) for all three optimized scores (colors) within each panel. One caveat to this interpretation is the bi-modality and long tail of the out-of-sample score for the one-stage combination (first row) optimizing the $\mathrm{CS}_{<10\%} $ (teal) and $\mathrm{CS}_{<20\%}$ (purple), especially in Panels B and C.

\begin{figure}[t]
\includegraphics[width=\textwidth]{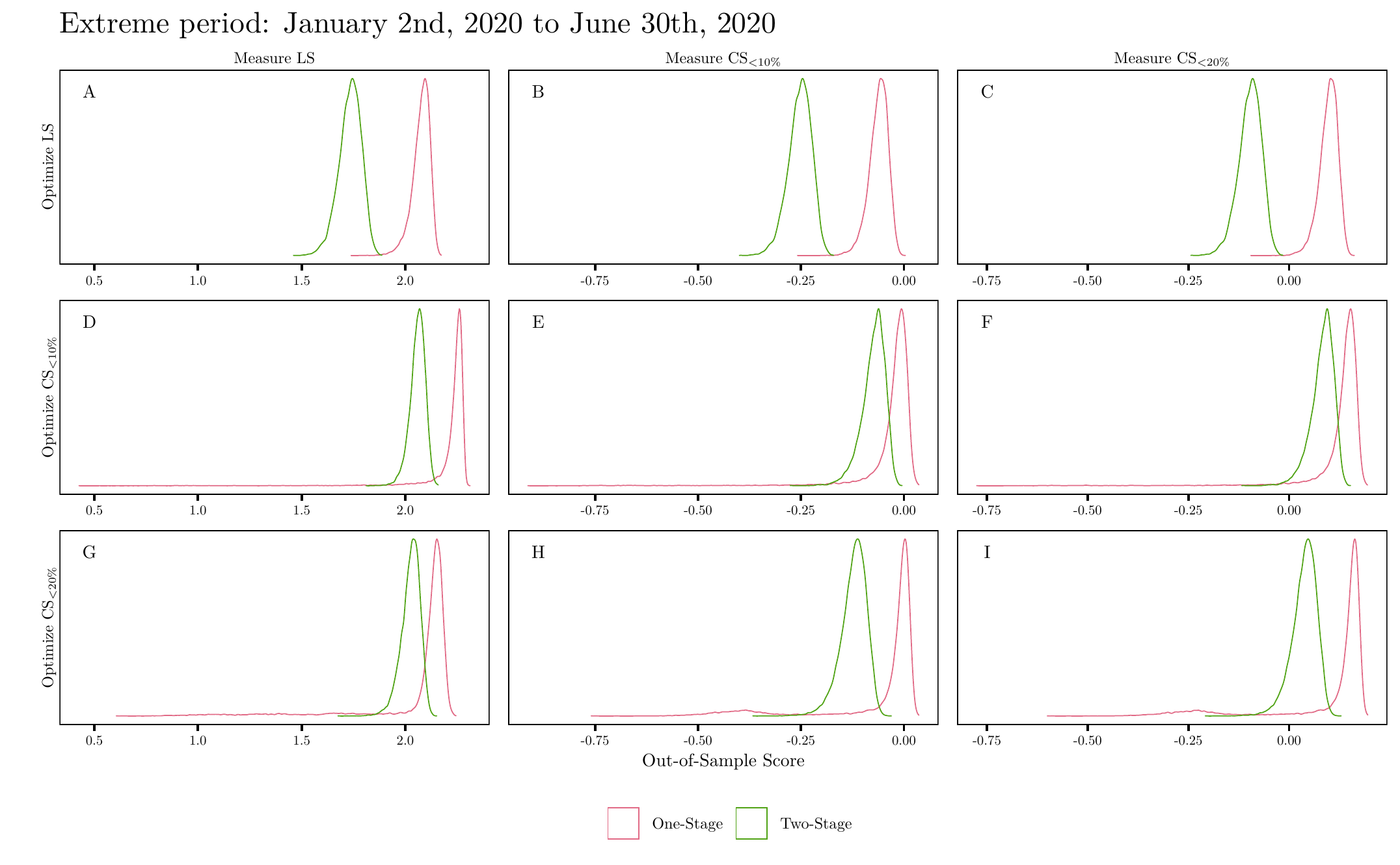}
\caption{Sampling distributions capturing the variability of average extreme-period scores (columns) due to the in-sample one-stage estimation (red) and two-stage estimation (green) of forecast combinations of S\&P500 returns according to a variety of in-sample scores (rows).}
\label{fig:emp4}
\end{figure}

\begin{figure}[t]
\includegraphics[width=\textwidth]{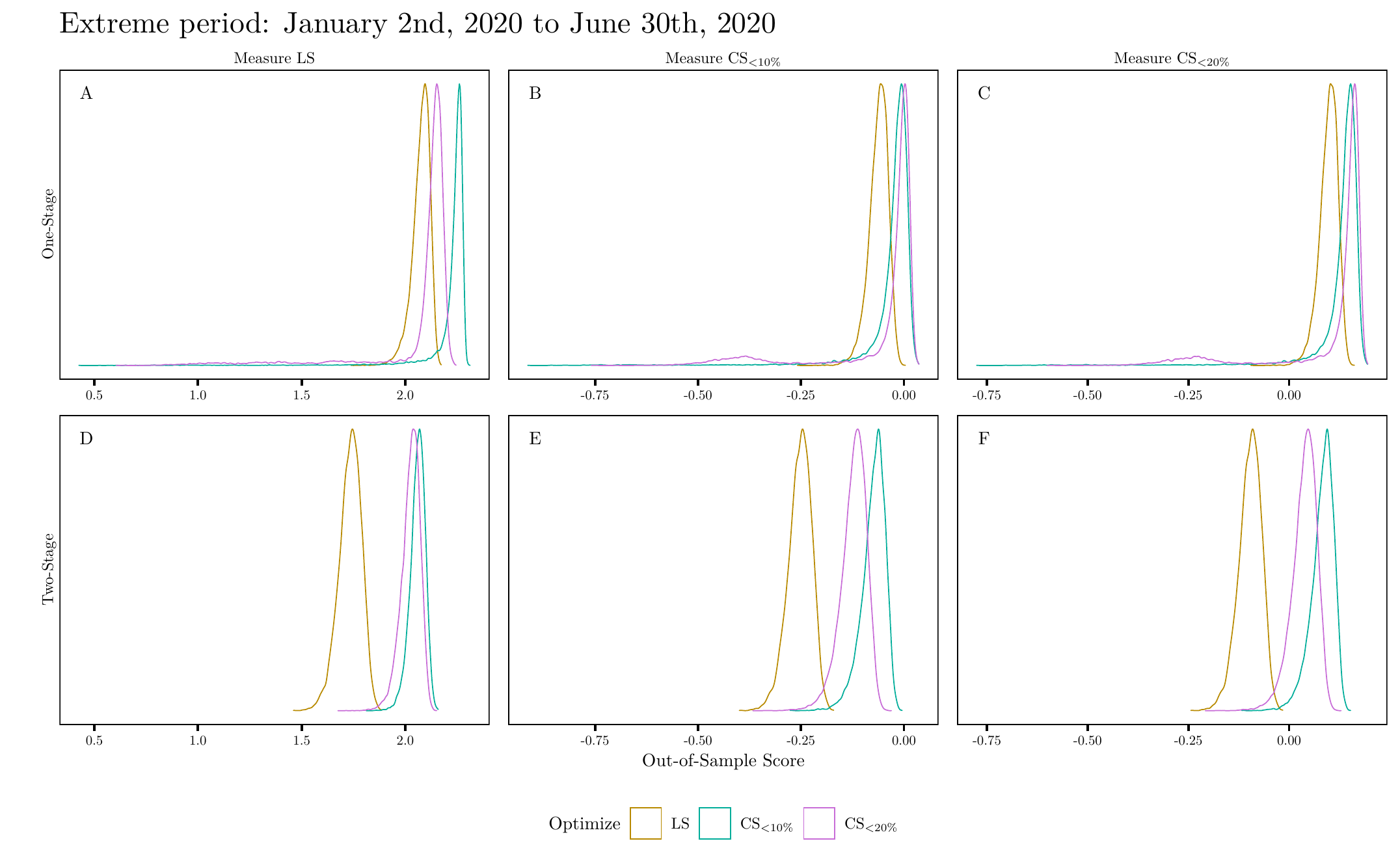}
\caption{Sampling distributions capturing the variability of average extreme-period scores (columns) due to the in-sample one-stage estimation (first row) and two-stage estimation (second row) of forecast combinations of S\&P500 returns according to a variety of in-sample scores (colors).}
\label{fig:emp5}
\end{figure}

Remarkably, the dominance of the one-stage forecast combinations over their two-stage counterparts persists as we move from the overall period to the extreme period, and is now uniform. This is reflected in Figure \ref{fig:emp4}, which shows a higher sampling distribution for the average out-of-sample scores of one-stage combinations (red) relative to their two-stage counterparts (green), regardless of how out-of-sample performance is measured (columns) or which in-sample score is optimized (rows). The asymptotic superiority of the one-stage approach implied by Theorem \ref{thm:one} is therefore reflected in our finite-sample results even as volatility increases and the size of the out-of-sample period declines.

Finally, Figure \ref{fig:emp5} addresses, for the extreme period, the effect of optimizing for the different scores (colors). There are two key conclusions to draw from a comparison of the results for the extreme period in Figure \ref{fig:emp5} and the results for the overall period in Figure \ref{fig:emp3}. First, relative to the corresponding panel in Figure \ref{fig:emp3}, the sampling distributions for the out-of-sample scores for the extreme period in Figure \ref{fig:emp5} are both lower (that is, further to the left) and broader than those for the overall period, which is reflected primarily in the different scales on the $x$-axes of corresponding panels. Second, the LS-optimizing forecast combinations (gold) now underperforms in terms of all out-of-sample measures (columns), relative to the $\mathrm{CS}_{<10\%}$- and $\mathrm{CS}_{<20\%}$-optimizing forecast combinations (teal and purple, respectively), whether optimizing in a one-stage (first row) or two-stage (second row) fashion. In the extreme period, optimizing according to a censored log score \textit{always} leads to a better out-of-sample performance, whichever measure of performance is used.

\section{Conclusion\label{sec:conclusion}}

In this paper, we have compared the forecast performance of a variety of methods for estimating the parameters of a forecast combination. Forecast performance is measured according to an expected out-of-sample scoring rule, and we take into account the sampling variability of this measure of forecast performance that arises due to sampling variability in the estimation of the forecast combination parameters -- parameters that underlie the constituent models as well as the parameters of the combination function. In this context, we analyze and compare the forecast performance of one- and two-stage forecast combinations produced according to either the score used to measure forecast performance, or some other score.

We have found via standard asymptotic arguments that for traditional two-stage forecast combinations, uncertainty in forecast performance is dominated by uncertainty in the estimation of the constituent models. Regarding the forecast combination puzzle, this lends support to optimal two-stage combinations over their equally weighted counterparts, since the limiting performance of the former is necessarily at least as high as the latter, with no added limiting sampling variability. Unfortunately, the standard practice in the forecast combinations literature is to hold the constituent models fixed when studying the forecast performance consequences of different estimation methods. Hence, this dominant source of uncertainty is often neglected in practice, including in studies of the forecast combination puzzle. Additionally, we show that producing one-stage forecast combinations according to a measure of forecast performance typically results in this measure having a higher limiting value and a lower limiting sampling variability, relative to two-stage forecast combinations. Alongside this result, we also provide formulae for producing estimates for the sampling distribution of the forecast performance of one- and two-stage
forecast combinations.

A simulation exercise confirms the theoretical results on all fronts, across multiple measures of performance, and for sample sizes typically observed in practice. In particular, knowing the limiting combination weights before observing the data does not appreciably change the forecast performance of the two-stage combination, provided that the performance measure and the optimization criterion coincide. Further, one-stage forecast combinations always beats two-stage forecast combinations as long as the sample size is large, even when assessed according to a scoring rule that differs from that used to produce the forecast combination.

In our empirical exercise, forecast combinations of S\&P500 returns continue to provide strong support for the benefit of one-stage forecast combinations over two-stage forecast combinations, most notably in the period of pandemic-induced high volatility in the first half of 2020. Moreover, in this latter period, both one-stage \textit{and} two-stage combinations optimized according to a tail-based scoring rule outperform the corresponding combinations optimized according to the log score, no matter what measure of out-of-sample performance is used.

We see several opportunities for further investigation. First, there is no doubt that this work has implications for other contexts where two-stage estimation is commonly employed, other than (distributional) forecast combinations. We expect that the advice given above on the source of forecast performance sampling variability in the two-stage cases, and the superior performance of one-stage forecast combinations, will apply equally to those contexts. Secondly, we are interested to discover the implications of these results for tests of predictive ability (e.g.\ \citealp{Diebold1995}; \citealp{Hansen2005}; \citealp{Giacomini2006}), especially within the context of the forecast combination puzzle. Thirdly, we would like to explore whether our estimates for the sampling distribution of forecast performance derived in Sections \ref{sec:implications} and \ref{subsec:empdatacomb} can be improved upon by a more sophisticated bootstrap-based methodology.

Finally, we leave forecasters with the following \textit{guidelines}:

\smallskip 

\noindent $\bullet$ Where possible, consider using one-stage forecast combinations instead of the standard two-stage approach. We also advocate for producing forecast combinations by optimizing the measure of forecast performance most important for the problem at hand, and for the horizon to be forecast by the combination.

\smallskip 

\noindent $\bullet$ When comparing the performance of competing forecast combinations, consider the impact that sampling variability in the parameter estimates has on the sampling variability of the forecast performance measure. In particular, it is important to consider the impact of sampling variability in \textit{all} estimated parameters of a combination, including those that govern the constituent models. For combinations produced in two-stages, sampling variability in the estimation of the constituent model parameters can dominate overall forecast performance; hence, ignoring this component of sampling variability may lead to misleading conclusions about the benefits, or otherwise, of optimizing the combination weights.

\smallskip 

\noindent $\bullet$ When forecasting during times of high volatility, consider using a scoring rule that focuses on forecast accuracy in the tails, rather than relying only on the log score.

\if0\blind
{
\section*{Acknowledgments}
We would like to thank Eric Eisenstat, Rob Hyndman, Mervyn Silvapulle and Farshid Vahid-Araghi for their thoughtful comments, which have been of great benefit to the quality of our article.

\section*{Disclosure Statement}
The authors report there are no competing interests to declare.
}\fi

{
\footnotesize
\bibliographystyle{ECA_jasa}
\bibliography{library}
}

\newpage

\pagenumbering{arabic}

\if0\blind
{
  \spacingset{1}
  \setcounter{footnote}{0}
  \renewcommand{\thefootnote}{\fnsymbol{footnote}}
  \maketitle
  \footnotetext[1]{\funding}
  \footnotetext[2]{\ebsaffiliation}
  \footnotetext[3]{\mdaffiliation}
  \footnotetext[4]{\correspondingauthor}
  \renewcommand{\thefootnote}{\arabic{footnote}}
  \spacingset{1.5}
} \fi

\if1\blind
{
  \bigskip
  \bigskip
  \bigskip
  \begin{center}
    {\Large\bf The Impact of Sampling Variability on Estimated Combinations of Distributional Forecasts}
  \end{center}
  \smallskip
} \fi

\appendix

\section*{\centering Appendix: Theoretical Results}
\setcounter{section}{1}

The appendix contains technical details supporting the theory, simulations and empirical exercises developed in the main paper. \appendixname\ \ref{subsec:regularity} provides standard regularity conditions assumed in the theorems of Section \ref{sec:implications}. In \appendixname\ \ref{subsec:gmm}, we provide a GMM representation of two-stage forecast combinations and provide expressions for the asymptotic sampling distributions of the one- and two-stage parameter estimates. The structure of the asymptotic covariance matrix for the two-stage parameters are further explored in \appendixname\ \ref{subsec:gmmasycov}. The proofs for Theorems \ref{thm:one} and \ref{thm:two} can be found in \appendixname\ \ref{subsec:proofs}.

\if0\blind
{
\newpage
} \fi

\setcounter{equation}{0}
\renewcommand{\theequation}{\thesection.\arabic{equation}}


\subsection{Regularity Conditions\label{subsec:regularity}}

\begin{assumption}
\label{ast:exp}
The parameter space $\Theta$ is compact. There exists a real-valued deterministic function $\mathcal{S}_0(\theta) \equiv \mathcal{S}_{0}(\eta, \gamma)$, continuous on $\Theta$, such that the following are satisfied.

\begin{enumerate}
\item $\sup_{\theta\in\Theta} \| \mathcal{S}_{n}(\theta) - \mathcal{S}_0(\theta) \| = o_p(1)$.
\item There exists a unique vector $\theta^0 \coloneqq \argmax_{\theta \in \Theta} \mathcal{S}_0(\theta)$.
\end{enumerate}
\end{assumption}

\begin{assumption}
\label{ast:exp2}
For each $j=1, \ldots, K$, there exists a real-valued deterministic function $\mathcal{S}_0(\gamma_j)$, continuous on $\Gamma_j$,
such that the following are satisfied.

\begin{enumerate}
\item $\sup_{\gamma_j\in\Gamma_j} \| \mathcal{S}_{n}(\gamma_j) - \mathcal{S}_0(\gamma_j) \| = o_p(1)$.
\item There exists a unique vector $\gamma^{\star}_j \coloneqq \argmax_{\gamma_j \in \Gamma_j} \mathcal{S}_0(\gamma_j)$.
\item There exists a unique vector $\eta^{\star} \coloneqq \argmax_{\eta \in \mathcal{E}} \mathcal{S}_0(\eta, \gamma^{\star})$, where $\gamma^{\star} \coloneqq [ \gamma^{\star \prime}_{1}\ \cdots\ \gamma^{\star \prime}_{K} ]^{\prime}$.
\end{enumerate}
\end{assumption}

\begin{remark}
Assumptions \ref{ast:exp} and \ref{ast:exp2} give the requisite regularity conditions to ensure consistency of $\hat{\theta_{n}}$ and $\tilde{\theta_{n}}$; the proofs of which are standard and hence omitted for brevity.
\end{remark}

Furthermore, we assume the following, which, when satisfied, will allow us to deduce the asymptotic distributions of $\hat\theta_{n}$ and $\tilde{\theta_{n}}$.

\begin{assumption}
\label{ast:dist}
The following are satisfied.

\begin{enumerate}
\item $\theta^0, \theta^{\star} \in \mathrm{Int}(\Theta)$.
\item The functions $\mathcal{S}_0(\theta)$ and $\mathcal{S}_{n}(\theta)$ are twice continuously differentiable on $\text{Int}(\Theta)$.
\item The functions $q_{n}(\eta,\gamma) \coloneqq \partial \mathcal{S}_{n}(\eta,\gamma)/\partial\theta$ and $\tilde{g}_{n}(\eta,\gamma)$, defined in \appendixname\ \ref{subsec:gmm}, exist and satisfy the following.
\begin{enumerate}
\item $\sqrt{n} q_{n}(\eta^0, \gamma^0) \Rightarrow N(0, V^0)$ and $\sqrt{n} \tilde{g}_{n}(\eta^\star, \gamma^\star) \Rightarrow N(0, V^\star)$.

\item There exist matrices $M^0(\theta) \coloneqq \partial^2 \mathcal{S}_0(\theta)/\partial \theta \partial \theta^{\prime}$ and $M^\star(\theta)$, nonsingular at $\theta^0$ and $\theta^\star$, respectively, such that for any nonnegative sequence $\delta_n \to 0$, 
\begin{align*}
\sup_{\| \theta - \theta^\star \| \leq \delta_n} \| \partial \tilde{g}_{n}(\theta)/\partial\theta^{\prime} - M^\star(\theta) \| &= o_p(1), \\
\sup_{\| \theta - \theta^0 \| \leq \delta_n} \| \partial {q}_{n}(\theta) / \partial \theta^{\prime}-M^0(\theta)\| &= o_p(1).
\end{align*}
\end{enumerate}
\end{enumerate}
\end{assumption}

\subsection{GMM Representation of One- and Two-Stage Forecast Combinations\label{subsec:gmm}}

Consider the derivatives $g_{n}(\eta, \gamma) \coloneqq  \partial \mathcal{S}_{n}(\eta, \gamma)/\partial \eta$ and $m_{j,n}(\gamma_j) \coloneqq \frac{1}{n} \sum_{t = 1}^{n} \partial S(F^{\gamma_j}_{j,t}, Y_{t})/\partial \gamma_j$, and stack $m_{j,n}(\gamma_j)$ to obtain $m_{n}(\gamma) \coloneqq [ m_{1,n}(\gamma_1)^{\prime}\ \cdots\ m_{K,n}(\gamma_K)^{\prime} ]^{\prime}$. Since $\tilde{\eta}_{n}$ maximizes $\mathcal{S}_{n}(\eta, \tilde{\gamma}_n)$ and $\tilde{\gamma}_{j,n}$ maximizes $\frac{1}{n} \sum_{t = 1}^{n} S(F^{\gamma_j}_{j,t}, Y_{t})$, we have that, under standard regularity conditions, $\eta = \tilde{\eta}_{n}$ solves 
\begin{equation}
g_{n}(\eta, \tilde{\gamma}_{n}) = 0,  \label{eqn:eta2s}
\end{equation}
with probability approaching one, and $\gamma = \tilde{\gamma}_{n}$ solves 
\begin{equation}
m_{n}(\gamma) = \begin{bmatrix} m_{1,n}(\gamma) \\ \vdots \\ m_{K,n}(\gamma) \end{bmatrix} = 0,  \label{eqn:gamma2s}
\end{equation}
with probability approaching one. Stacking $g_{n}$ and $m_{n}$ as $\tilde{g}_{n}$, we see that \eqref{eqn:eta2s} and \eqref{eqn:gamma2s} form the joint moment equation 
\begin{equation*}
\tilde{g}_{n}(\tilde{\eta}_{n}, \tilde{\gamma}_{n}) \coloneqq \begin{bmatrix} g_{n}(\eta, \gamma) \\ m_{n}(\gamma) \end{bmatrix} \bigg{|}_{\eta = \tilde{\eta}_{n}, \gamma = \tilde{\gamma}_{n}}.
\end{equation*}
The two-stage estimator in \eqref{eqn:est2s2} is therefore a GMM estimator based on the moment function $\tilde{g}_{n}$. See \citesupp{Newey1994} for further details regarding GMM.

Now let $V^0$ and $V^{\star}$ be the asymptotic covariance matrices of the scaled gradient vectors $\sqrt{n} \partial \mathcal{S}_n(\theta^0) / \partial \theta$ and $\sqrt{n} \tilde{g}_n(\eta^{\star}, \gamma^{\star})$ of the one- and two-stage estimators, respectively. Given the assumptions in \appendixname\ \ref{subsec:regularity}, we have, via the usual arguments, 
\begin{align}
\sqrt{n}(\hat{\theta}_n - \theta^0) & \Rightarrow N(0, W^0),  \notag \\
W^0 & \coloneqq [M^0]^{-1} V^0 [M^0]^{{-1}^{\prime}},  \label{eqn:onestagew}
\end{align}
and 
\begin{align*}
\sqrt{n}(\tilde{\theta}_n - \theta^{\star}) & \Rightarrow N(0, W^{\star}), \\
W^{\star} & \coloneqq [M^{\star}]^{-1} V^{\star} [M^{\star}]^{{-1}^{\prime}},
\end{align*}
where $M^0 \coloneqq \mathrm{plim}_n \partial q_n(\theta^0) / \partial \theta^{\prime}$ and $W^{\star} \coloneqq \mathrm{plim}_n \partial \tilde{g}_n(\theta^{\star}) / \partial \theta^{\prime}$. The asymptotic covariance matrix $W^{\star}$ of the two-stage estimator $\tilde{\theta}_n$ has a specific structure due to its two-stage nature. See \appendixname\ \ref{subsec:gmmasycov} for details.

\subsection{Structure of the Asymptotic Covariance Matrix of Two-Stage Forecast Combination Parameter Estimates\label{subsec:gmmasycov}}

Recall that the asymptotic sampling distribution of the two-stage forecast combination parameters estimates is 
\begin{align*}
\sqrt{n}(\tilde{\theta}_n - \theta^{\star}) & \Rightarrow N(0, W^{\star}), \\
W^{\star} & \coloneqq [M^{\star}]^{-1} V^{\star} [M^{\star}]^{-1 \prime},
\end{align*}
where $V^{\star}$ is the asymptotic covariance matrix of the normalized gradient $\sqrt{n} \tilde{g}_{n}(\eta^{\star}, \gamma^{\star})$, and $M^{\star} \coloneqq \mathrm{plim}_n \partial \tilde{g}_{n} (\theta^{\star}) / \partial \theta$. See \appendixname\ \ref{subsec:gmm} for the definition of $\tilde{g}_n$, and for the definitions of $g_n$ and $m_n$, used below.

The matrix $M^\star$ has a particular structure, namely, denoting $G_{\eta} \coloneqq \mathrm{plim}_n \partial g_{n}(\theta^\star) / \partial \eta$, $G_{\gamma} \coloneqq \mathrm{plim}_n \partial g_{n}(\theta^{\star}) / \partial \gamma$, and $M_{\gamma} \coloneqq \mathrm{plim}_n \partial m_{n}(\gamma^{\star}) / \partial \gamma^{\prime}$, we have 
\begin{equation*}
M^{\star} \coloneqq
\begin{bmatrix}
G_{\eta} & G_{\gamma} \\ 
0 & M_{\gamma}
\end{bmatrix},
\end{equation*}
so that 
\begin{equation*}
[M^\star]^{-1} \coloneqq 
\begin{bmatrix}
G^{-1}_{\eta} & -G_{\eta}^{-1}G_{\gamma} M^{-1}_{\gamma} \\ 
0 & M^{-1}_{\gamma} \end{bmatrix}.
\end{equation*}
Using this, we can conclude that (see, e.g., Theorem 6.1 in \citealpsupp{Newey1994} for details) 
\begin{gather}
\sqrt{n}(\tilde{\eta}_{n} - \eta^\star) \Rightarrow N\left(0, W^{\star}_{\eta}\right), \notag \\
W^{\star}_{\eta} \coloneqq G_{\eta}^{-1} V^{\star}_{\eta} G_{\eta}^{-1\prime}, \notag \\
V^{\star}_{\eta} \coloneqq \lim_{n \to \infty} \text{Var} \left[ \sqrt{n} \left\{ g_{n}(\theta^{\star}) + G_{\gamma} M^{-1}_{\gamma} m_{n}(\gamma^{\star}) \right\} \right]. \label{eqn:twostagev}
\end{gather}

Moreover, comparing \eqref{eqn:onestagew} and \eqref{eqn:twostagev}, we see that ignoring the first stage in estimating the asymptotic distribution of the second-stage estimator, $\tilde{\eta}_n$, is equivalent to assuming $\sqrt{n} G_{\gamma} M^{-1}_{\gamma} m_n(\gamma^{\star}) = o_p(1)$. Thus, ignoring the first stage leads to valid second-stage inference if $G_{\gamma} = 0$. Assuming we may differentiate inside the expectation, the implicit function theorem ensures the existence of a unique function $\eta^0(\gamma) \coloneqq \argmax_{\eta \in \mathcal{E}} \mathcal{S}_0(\eta, \gamma)$, defined in a neighborhood of $\gamma^{\star}$, such that 
\begin{equation*}
\frac{\partial \eta^0(\gamma^{\star})}{\partial \gamma^{\prime}} = -G_{\eta}^{-1} G_{\gamma}.
\end{equation*}
Now consider the condition $\partial \eta^{0} (\gamma^{\star}) / \partial \gamma^{\prime}$ = 0. Under our regularity conditions, this is equivalent to the expected-score-maximizing second-stage parameter vector $\eta^{\star} = \eta^0(\gamma)$ being constant in a neighborhood of $\gamma^{\star}$ (Theorem 6.2, \citealpsupp{Newey1994}). In this case, we would have $G_{\gamma} = 0$, and ignoring the first stage would lead to valid second stage inference.

In general, we see that ignoring first-stage estimation will deliver invalid inference on second-stage parameters, excepting cases where the optimal value of the first-stage parameters, $\gamma^{\star}$, has no impact on the optimal value of the combination weights, $\eta^{\star}$.

\subsection{Proofs\label{subsec:proofs}}

\begin{proof}[Proof of Theorem \ref{thm:one}]
The first implication follows directly from Assumptions \ref{ast:exp} and \ref{ast:exp2}, and the definition of $\tilde\theta_{n}$ and $\hat{\theta}_{n}$. The second result is a consequence of the two-stage nature of $\tilde\theta_{n}$ and a first-order Taylor series expansion. In particular,
\begin{align*}
\sqrt{n} & \{ \mathcal{S}_0(\tilde{\eta}_{n},\tilde{\gamma}_{n}) - \mathcal{S}_0(\eta^{\star},\gamma^{\star}) \} \\
= & \left[{\partial \mathcal{S}_0(\eta^{\star}, \gamma^{\star})}/{\partial \eta}\right]^{\prime} \sqrt{n}(\tilde{\eta}_{n} - \eta^{\star}) \\
& + \left[{\partial \mathcal{S}_0(\eta^{\star}, \gamma^{\star})}/{\partial \gamma}\right]^{\prime} \sqrt{n} (\tilde{\gamma}_{n} - \gamma^{\star}) \\
& + \left[{\partial \mathcal{S}_0(\bar{\theta}_{n})}/{\partial \theta} - {\partial \mathcal{S}_0(\theta^{\star})}/{\partial \theta} \right]^{\prime} \sqrt{n}(\tilde{\theta}_{n} - \theta^{\star}) \\
= & [{\partial \mathcal{S}_0(\eta^\star, \gamma^\star)}/{\partial \gamma}]^{\prime} \sqrt{n}(\tilde{\gamma}_{n} - \gamma^{\star}) + o_p(1)
\end{align*}
for $\| \bar\theta_{n} - \theta^{\star} \| \leq \| \tilde{\theta}_{n} - \theta^{\star} \| = \mathcal{O}_p(n^{-1/2})$. By Assumptions \ref{ast:exp2}.3, \ref{ast:dist}.1 and \ref{ast:dist}.2, $\partial \mathcal{S}_0(\eta^{\star}, \gamma^{\star})/\partial \eta = 0$ and the term on the second line disappears. By Assumption \ref{ast:dist}.2, the term on the fourth line converges to zero in probability. Similarly,
\begin{align*}
\sqrt{n} & \{ \mathcal{S}_0(\eta^{\star}, \tilde{\gamma}_{n}) - \mathcal{S}_0(\eta^{\star}, \gamma^{\star}) \} \\
= & \left[{ \partial \mathcal{S}_0(\eta^{\star}, \bar{\gamma}_{n})}/{\partial \gamma} \right]^{\prime} \sqrt{n}(\tilde{\gamma}_{n} - \gamma^{\star}) \\
= & \left[{ \partial \mathcal{S}_0(\eta^{\star}, \gamma^{\star})}/{\partial \gamma}\right]^{\prime}\sqrt{n}(\tilde{\gamma}_{n}-\gamma^{\star}) + o_p(1)
\end{align*}
for $\lVert \bar{\gamma}_n - \gamma^{\star} \rVert \leq \lVert \tilde{\gamma}_n - \gamma^{\star} \rVert = \mathcal{O}_p(n^{-1/2})$. Subtracting the two expansions yields the result.
\end{proof}

\begin{proof}[Proof of Theorem \ref{thm:two}, Part 1.]
Recall that $\sqrt{n}(\tilde{\theta}_{n} - \theta^{\star}) \Rightarrow N(0, W^{\star})$, and apply the first-order delta method, using Assumption \ref{ast:dist}.2 to control the error. The asymptotic covariance matrix follows from $\partial \mathcal{S}_0(\eta^{\star}, \gamma^{\star})/\partial \eta = 0$ and some block matrix manipulation.
\end{proof}

\begin{proof}[Proof of Theorem \ref{thm:two}, Part 2.]
By Assumptions \ref{ast:exp}.2, \ref{ast:dist}.1 and \ref{ast:dist}.2, we have $\partial \mathcal{S}_0(\theta^0) / \partial \theta = 0$. Now recall that $\sqrt{n}(\hat{\theta}_{n} - \theta^0) \Rightarrow X$, so we apply the second-order delta method using a second-order Taylor series expansion of $\mathcal{S}_0(\hat{\theta}_n)$ around $\theta = \theta^0$. Use Assumption \ref{ast:dist}.3 (b) to control the error.
\end{proof}

{
\footnotesize
\bibliographystylesupp{ECA_jasa}
\bibliographysupp{library}
}

\end{document}

%% file: figure/TBL1.tex
\begin{tabular}{lllllll}
  \rowcolor{Gray} \multicolumn{7}{c}{Panel A -- Overall period: January 3rd, 2017 to December 31st, 2021} \\
  & & & & \multicolumn{3}{c}{Out-of-Sample Score} \\
  \cmidrule(lr){5-7}
  & & & & LS & $\mathrm{CS}_{<10\%}$ & $\mathrm{CS}_{<20\%}$ \\
  \hline
  \multirow{12}{*}{\rotatebox[origin=c]{90}{In-Sample Score}} & LS & One-Stage & Average & \textbf{3.289} & \textbf{-0.003864} & 0.1268 \\ 
  & & & 95\% CI & \textbf{(3.278, 3.295)} & \textbf{(-0.009395, -0.0006305)} & \textbf{(0.1203, 0.1313)} \\ 
  & & Two-Stage & Average & 3.186 & -0.04076 & 0.07958 \\ 
  & & & 95\% CI & (3.172, 3.198) & (-0.04779, -0.03632) & (0.07239, 0.08438) \\ 
  & $\mathrm{CS}_{<10\%}$ & One-Stage & Average & 3.206 & -0.004871 & 0.1243 \\ 
  & & & 95\% CI & (3.134, 3.244) & (-0.06344, -0.0009392) & (0.07003, 0.1311) \\ 
  & & Two-Stage & Average & 2.814 & -0.009845 & \textbf{0.1309} \\ 
  & & & 95\% CI & (2.698, 2.924) & (-0.01797, -0.005612) & (0.1226, 0.1350) \\ 
  & $\mathrm{CS}_{<20\%}$ & One-Stage & Average & 3.180 & -0.003953 & 0.1263 \\ 
  & & & 95\% CI & (3.024, 3.241) & (-0.06838, 0.0002676) & (0.06320, 0.1360) \\ 
  & & Two-Stage & Average & 2.971 & -0.02151 & 0.1171 \\ 
  & & & 95\% CI & (2.906, 3.035) & (-0.03008, -0.01702) & (0.1083, 0.1219) \\ 
  \hline
  \rowcolor{Gray} \multicolumn{7}{c}{Panel B -- Extreme period: January 2nd, 2020 to June 30th, 2020} \\
  & & & & \multicolumn{3}{c}{Out-of-Sample Score} \\
  \cmidrule(lr){5-7}
  & & & & LS & $\mathrm{CS}_{<10\%}$ & $\mathrm{CS}_{<20\%}$ \\
  \hline
  \multirow{12}{*}{\rotatebox[origin=c]{90}{In-Sample Score}} & LS & One-Stage & Average & 2.085 & -0.05693 & 0.1028 \\ 
  & & & 95\% CI & (1.972, 2.139) & \textbf{(-0.1152, -0.02416)} & \textbf{(0.04506, 0.1348)} \\ 
  & & Two-Stage & Average & 1.739 & -0.2501 & -0.09458 \\ 
  & & & 95\% CI & (1.613, 1.827) & (-0.3159, -0.2043) & (-0.1605, -0.04892) \\ 
  & $\mathrm{CS}_{<10\%}$ & One-Stage & Average & \textbf{2.263} & -0.005679 & 0.1527 \\ 
  & & & 95\% CI & (1.898, 2.284) & (-0.2210, 0.01739) & (-0.06324, 0.1753) \\ 
  & & Two-Stage & Average & 2.068 & -0.06809 & 0.08888 \\ 
  & & & 95\% CI & \textbf{(1.975, 2.121)} & (-0.1438, -0.03125) & (0.01391, 0.1252) \\ 
  & $\mathrm{CS}_{<20\%}$ & One-Stage & Average & 2.166 & \textbf{0.005156} & \textbf{0.1656} \\ 
  & & & 95\% CI & (1.198, 2.201) & (-0.4517, 0.02035) & (-0.2926, 0.1801) \\ 
  & & Two-Stage & Average & 2.033 & -0.1167 & 0.04320 \\ 
  & & & 95\% CI & (1.922, 2.099) & (-0.1957, -0.07268) & (-0.03509, 0.08658) \\ 
  \hline
\end{tabular}